\documentclass[a4paper,12pt]{amsart}

\usepackage[utf8]{inputenc}

\usepackage{geometry}
\usepackage{geometry}
\usepackage{amsmath,amsfonts,amsthm, amssymb,mathtools}
\usepackage{hyperref, color}

\usepackage{soul}

\newcommand{\R}{\mathbb{R}}

\newcommand{\bZ}{\mathbb{Z}}

\newcommand{\hA}{\widehat{A}}
\newcommand{\AhA}{A\times\widehat{A}}

\newcommand{\cE}{\mathcal{E}}

\newcommand{\cH}{\mathcal{H}}
\newcommand{\ket}[1]{|#1\rangle}
\newcommand{\bra}[1]{\langle #1|}
\newcommand{\braket}[2]{\langle #1|#2\rangle}
\DeclareMathOperator{\Tr}{Tr}

\let\originalleft\left
\let\originalright\right
\renewcommand{\left}{\mathopen{}\mathclose\bgroup\originalleft}
\renewcommand{\right}{\aftergroup\egroup\originalright}

\newtheorem{theorem}{Theorem}[section]
\newtheorem{lemma}[theorem]{Lemma}
\newtheorem{proposition}[theorem]{Proposition}
\newtheorem{remark}[theorem]{Remark}
\newtheorem{definition}[theorem]{Definition}
\newtheorem{corollary}[theorem]{Corollary}
\newtheorem{example}[theorem]{Example}
\numberwithin{equation}{section}

\title[The wave function of stabilizer states]{The wave function of stabilizer states and the Wehrl conjecture}

\author[Fabio Nicola]{Fabio Nicola}
\address{Dipartimento di Scienze Matematiche, Politecnico di Torino, Corso Duca degli Abruzzi 24, 10129 Torino, Italy}
\email{fabio.nicola@polito.it}

\date{}

\begin{document} 
 
\begin{abstract}
We focus on quantum systems represented by a Hilbert space $L^2(A)$, where $A$ is a locally compact Abelian group that contains a compact open subgroup.
We examine two interconnected issues related to Weyl-Heisenberg operators. First, we provide a complete and elegant solution to the problem of describing the stabilizer states in terms of their wave functions, an issue that arises in quantum information theory. Subsequently, we demonstrate that the stabilizer states are precisely the minimizers of the Wehrl entropy functional, thereby resolving the analog of the Wehrl conjecture for any such group. Additionally, we construct a moduli space for the set of stabilizer states, that is, a parameterization of this set, that endows it with a natural algebraic structure, and we derive a formula for the number of stabilizer states when $A$ is finite. Notably, these results are novel even for finite Abelian groups. 
\end{abstract}

\subjclass[2010]{81P45, 81S30, 43A25}
\keywords{Stabilizer states, locally compact Abelian groups, Weyl-Heisenberg operators, Wehrl entropy}

\maketitle

\section{Introduction and discussion of the main results}
Consider a Hilbert space $\cH$ together with a ``representation", i.e. a unitary operator $U:\cH\to L^2(A)$, where $A$ is a locally compact Abelian (LCA) group. For $\ket{\psi}\in\cH$ we denote by $\psi=U \ket{\psi}\in L^2(A)$ its {\it wave function}.
 We denote by $\hA$ the group of continuous characters of $A$, namely the continuous homomorphisms $\xi:A\to U(1)$ (the multiplicative group of complex numbers of modulus $1$). For $x\in A$, $\xi\in\hA$ we denote by $\xi(x)\in U(1)$ the value of $\xi$ at $x$. In addition, we denote the group laws in $A$ and $\hA$ additively. We interpret $A$ as the classical configuration space, and the group product $\AhA$ as the classical phase space associated with $\cH$. 
 
Given $(x,\xi)\in\AhA$, we define the phase space shift $\pi(x,\xi):L^2(A)\to L^2(A)$  by
\begin{equation}\label{eq tfshift}
\pi(x,\xi)
\psi(y)=\xi(y) \psi(y-x)\qquad y\in A,\ \psi\in L^2(A). 
\end{equation}
We also define the Weyl-Heisenberg operators $W(x,\xi):\cH\to\cH$ as 
\begin{equation}\label{eq weylop}
W(x,\xi)= U^\dagger \pi(x,\xi)U\qquad (x,\xi)\in\AhA.
\end{equation}
These operators are widely used in both mathematics and physics, and their importance can hardly be overrated (see, e.g., \cite{degosson2011symplectic,degosson2021quantum,grochenig_book,luef2009quantum}). For a spinless particle in $\R^d$, in the Schr\"odinger representation, hence $A=\R^d$, they are generated by the usual position and momentum operators. The following is an illustration in a finite-dimensional context.

\begin{example}\label{exa quantum}
Suppose that $\cH$ has finite dimension, and let $\{\ket{x}\}_{x\in A}$ be an orthonormal basis, labeled by elements of the (finite) Abelian group $A$. This is a common framework for most systems occurring in quantum information, where such a fixed basis is referred to as the computational basis. For example, in the case of $n$ qudits, we have $A=\bZ_d^n$. The product $A=\bZ_{d_1}\times...\times \bZ_{d_n}$ arises similarly for multipartite systems; see \cite{basile2024weyl,nest2011} and the references therein. Observe that every finite Abelian group is isomorphic to the direct sum of finite cyclic groups and, therefore, occurs in this way. Here, the representation $U:\cH\to L^2(A)$, by definition, maps the vector $\ket{\psi}$ into its wave function $\psi(x):=\langle x|\psi\rangle$. 

For $(x,\xi)\in\AhA$, the Weyl-Heisenberg operator $W(x,\xi):\cH\to\cH$ acts on the computational basis as
\begin{equation}
W(x,\xi) \ket{y}=\xi(y+x)\ket{y+x} \qquad y\in A. 
\end{equation}

In the case of one qubit, with $A=\mathbb{Z}_2=\{0,1\}$, the shift operator $W(1,0)$ is known as the Pauli gate $X$, while the operator $W(0,1)$ is the Pauli gate $Z$; see \cite{basile2024weyl,bertlmann2008bloch,gross2006hudson,hostens2005stabilizer,weyl1927quantenmechanik}, the monographs \cite{holevo2011,holevo2019quantum},  \cite[Section 10.5]{nielsen2010quantum} and the references therein.
\end{example}
In this note we study two interconnected problems related to the Weyl-Heisenberg operators, that is, we offer a thorough and elegant solution to the problem of describing the so-called stabilizer states in terms of their wave functions, and we show that such states are exactly the minimizers of the Wehrl entropy functional. 
\subsection{Stabilizer states} 
Given $A$ as above, consider the corresponding {\it Heisenberg group}, defined as
\begin{equation}\label{eq heisen}
\mathbb{H}_A=\{u\, W(x,\xi):\ u\in U(1),\ (x,\xi)\in\AhA\}.
\end{equation}
It is a noncommutative locally compact and Hausdorff topological group, endowed with the topology induced by the natural bijection $\mathbb{H}_A\simeq U(1)\times\AhA$ (equivalently, the strong operator topology of bounded operators in $\mathcal{H}$; see, e.g., \cite[Lemma 1]{igusa68}). When $A$ is finite, say of cardinality $N$, we also consider the subgroup of $\mathbb{H}_A$, known as the {\it Pauli group}, defined as
 \begin{equation}\label{eq pauli}
\mathbb{P}_A=\{\zeta^k W(z):\ k\in\mathbb{Z},\ z\in\AhA\},
\end{equation}
where $\zeta=\exp(\pi i/N)$. 

It is well known that, if $A$ is finite, say of cardinality $N$, and $\mathcal{G}\subset \mathbb{P}_A$ is an Abelian subgroup of cardinality $N$ that does not contain multiples of the identity other than the identity itself, there exists exactly one state that is stabilized by $\mathcal{G}$, that is, a common eigenstate, with eigenvalues $1$, of all elements of $\mathcal{G}$. The study of these particular states began with \cite{gottesman1998heisenberg}, motivated by the realization that the most interesting states in quantum information theory are frequently more manageable via their stabilizer group $\mathcal{G}$ rather than through direct examination. Since then, these states have become an essential resource in the development of quantum error correction codes; see \cite[Chapter 10]{nielsen2010quantum} and the references therein. 

Not every state is a stabilizer state. Hence, one faces the problem of describing these states more explicitly in terms of their wave function. This issue was addressed in \cite{dehaene2003} and \cite{hostens2005stabilizer} for a system of $n$ qubits ($A=\bZ_2^n$) and $n$ qudits ($A=\bZ_d^n$) respectively. There, the authors presented a clever algorithm to construct the wave function of the stabilizer state from a given group $\mathcal{G}$ as above, requiring certain matrix manipulations in modular arithmetic, the search for a minimal set of generators of $\mathcal{G}$, the reduction to a Smith normal form and finally the resolution of a Diophantine system (see also
\cite{gross2006hudson} for the case where $A=\bZ^n_d$ and $d$ is an odd integer, in which the analysis is substantially simplified due to the fact that multiplication by $2$ acts as an isomorphism of $\bZ_d$, in this case).  
Although this method could theoretically be extended to finite Abelian groups, albeit with additional algebraic complexities, it seems that a more conceptual and systematic understanding has not yet been developed. In the first part of this note we will bridge this deficiency by introducing an elegant and general rule to derive the wave function of a stabilizer state from the group $\mathcal{G}$. 
Building on the foundational work \cite{segal1963,weil64} on Weyl-Heisenberg operators, we shall consider a general locally compact Abelian (LCA) group $A$ (containing a compact open subgroup; see below). Our analysis will be ``coordinate-free", based solely on the structure of $A$ as an LCA group, without depending on other choices, such as a set of generators. This could also offer significant practical advantages, as highlighted in \cite{nest2011}. We stress that the derived formula is novel even for finite Abelian groups and that our analysis is \textit{not} a modification of the approach discussed in the aforementioned references.

To properly state our results, we introduce some terminology. We recall that, according to \cite{weil64}, a \textit{character of second degree} (alias {\it quadratic form}) of an LCA group $A$ is a continuous function $h:A\to U(1)$ such that
\begin{equation}\label{eq hbeta}
h(x+y)=h(x)h(y)\,\beta(x)(y)\qquad x,y\in A
\end{equation}
for some continuous symmetric homomorphism $\beta:A\to\widehat{A}$; hence $\beta(x)(y)=\beta(y)(x)$ for $x,y\in A$. The set of characters of second degree of $A$ is easily seen to be a group that will be denoted by ${\rm Ch}_2(A)$. Moreover, we will denote by ${\rm Sym}(A)$ the group of continuous symmetric homomorphisms $\beta$ as above. It is clear that $\hA\subset {\rm Ch}_2(A)$ and it is known that ${\rm Ch}_2(A)/ \hA\simeq {\rm Sym}(A)$ (algebraic isomorphism).

We also say that a  function $h:A\to \mathbb{C}$ is a {\it subcharacter of second degree} of $A$ if there exists a compact open subgroup $H\subset A$ such that $h(x)=0$ for $x\in A\setminus H$ and the restriction of $h$ to $H$ is a character of second degree of $H$; we refer to Section \ref{sec 2} for details and for an explicit description (construction) of all characters of second degree of a finite Abelian group.

We now select a special class of state vectors.

\begin{definition}[$S$-state]\label{def state} We say that $\ket{\phi}\in\cH$ is an $S$-state vector if its wave function has the form 
\begin{equation}\label{eq stab phi}
\phi(x)=c\,h_0(x-y)
\end{equation}
for some $c\in\mathbb{C}\setminus\{0\}$, $y\in A$ and some subcharacter $h_0$ of second degree of $A$. We denote by $\mathfrak{S}$ the set of corresponding $S$-states, that is, states represented by an $S$-state vector. 
\end{definition}
%This terminology is justified by the fact that, as we will see in Section... (Proposition \ref{}) these states are precisely those whose characteristic function has support of minimum cardinality, namely $|A|$. 

We also give the following definition. We endow $\AhA$ with the Radon product measure, where the Haar measures on $A$ and $\hA$ are chosen so that the Fourier inversion formula holds with constant $1$. 
\begin{definition}\label{def gsubgroups}
Let $A$ be any LCA group and let $\mathbb{H}_A$ be the corresponding Heisenberg group. Consider the collection $\mathfrak{G}$ of the compact  subgroups $\mathcal{G}\subset \mathbb{H}_A$ such that the homomorphism 
\begin{equation}\label{eq 5g}
    \mathcal{G}\to\AhA\qquad u\,W(z)\mapsto z
\end{equation}
is injective (hence $\mathcal{G}$ is Abelian) and has an image of measure $1$ in $\AhA$.

We say that $\ket{\psi}\in\cH\setminus\{0\}$ is  a stabilizer state vector if there exists $\mathcal{G}\in \mathfrak{G}$, such that $\ket{\psi}$ is an eigenvector, with eigenvalue $1$, of all elements of $\mathcal{G}$. In this case, we will say that $\mathcal{G}$ stabilizes $\ket{\psi}$.
\end{definition} 
We refer to Remark \ref{rem osservazioni su mainteo} and Proposition \ref{pro no loss generality} below for a discussion and some results that justify the conditions that appear in the definition of the class $\mathfrak{G}$. Here, we limit ourselves to pointing out that $\mathfrak{G}$ is nonempty precisely when $A$ contains a compact open subgroup and that this scenario turns out to be the natural framework for a general stabilizer formalism (incidentally, if $A$ lacks any compact open subgroup, the statements below are vacously true). We also anticipate that, as a consequence of Theorem \ref{teo stab1} below, every $\mathcal{G}\in \mathfrak{G}$ stabilizes exactly one state.

By assumption, for any $\mathcal{G}\in\mathfrak{G}$ the homomorphism in \eqref{eq 5g} is injective. This implies that $\mathcal{G}$ is the graph of some function, that is, it has necessarily the form
\begin{equation}\label{eq 1}
\mathcal{G}=\{\alpha(z)W(z):\ z\in K\},
\end{equation}
for some compact  subgroup $K\subset\AhA$ of measure $|K|=1$ and some function $\alpha:K\to U(1)$.  Since $\mathcal{G}$ is compact, $\alpha$ is continuous (see Proposition \ref{pro struttura sottogruppo}). 

We can now state our first result. We recall that the annihilator of a subgroup $H\subset A$ is the subgroup $H^\bot\subset\hA$ given by \begin{equation}\label{eq annihilator}
H^\bot=\{\xi\in\hA:\ \xi(x)=1\ \textrm{for all}\ x\in H\}.
\end{equation}

\begin{theorem}\label{teo stab1}
Any $S$-state vector $\ket{\phi}$ (Definition \ref{def state}) is stabilized by only one $\mathcal{G}\in\mathfrak{G}$ (Definition \ref{def gsubgroups}), which can be obtained as follows. Let $\phi\in L^2(A)$ be the wave function of $\ket{\phi}$, hence of the form \eqref{eq stab phi} for some $y\in A$ and some subcharacter $h_0$ of second degree of $A$.
Let $H=\{x\in A:\ h_0(x)\not=0\}$, and let $h=h_0|_H\in{\rm Ch}_2(H)$ be associated with a continuous symmetric homomorphism $\beta\in{\rm Sym}(H)$ as in \eqref{eq hbeta} (hence, $\beta(x)\in\widehat{H}$ for $x\in H$). Then $\mathcal{G}$ has the form \eqref{eq 1} where \begin{equation}\label{eq K}
K=\{(x,\xi)\in\AhA:\ x\in H,\ \xi|_H=\beta(x)\}
\end{equation}
and
\begin{equation}\label{eq alpha}
\alpha(x,\xi)=\overline{h_0(-x)\xi(y)}\qquad (x,\xi)\in K.
\end{equation}

Conversely, given any subgroup $\mathcal{G}\in\mathfrak{G}$, there is only one corresponding stabilizer state. It is an $S$-state, that is, its wave function is given by the formula \eqref{eq stab phi}, where $h_0$ and $y$ are given as follows. 
\begin{itemize}
\item[(i)] With $\mathcal{G}$ as in \eqref{eq 1}, let $H$ be the image of $K$ under the canonical projection $\AhA\to A$. The map $\xi\mapsto\alpha(0,\xi)$ is a continuous character of $H^\bot$, and therefore there exists $y\in A$ such that \[
\xi(y)=\overline{\alpha(0,\xi)}\quad \textrm{for}\ \xi\in H^\bot.
\]
\item[(ii)] The function $h:H\to U(1)$ given by \[
h(-x)=\overline{\alpha(x,\xi)\xi(y)}\quad \text{for}\ x\in H\ \text{and any}\ \xi\in\hA\ \text{with}\ (x,\xi)\in K
\]
is well defined and belongs to ${\rm Ch_2}(H)$. 
\item[(iii)] Let $h_0:A\to\mathbb{C}$ be defined by $h_0(x)=h(x)$ for $x\in H$ and $h_0(x)=0$ for $x\in A\setminus H$. Then for these $h_0$ and $y$, the wave function $\phi(x):=h_0(x-y)$ defines the stabilizer state of $\mathcal{G}$. 
\end{itemize}

As a consequence, a state is a stabilizer state if and only if it is an $S$-state.
\end{theorem}
 Observe that, for practical purposes, if $\mathcal{G}$ is given in the form \eqref{eq 1}, the subgroup $H^\bot$ is also explicitly known, because $H^\bot=\{\xi\in\hA:\ (0,\xi)\in K\}$ (see Theorem \ref{teo iso sub}). See also Remark \ref{rem 3.5} for other observations regarding Theorem \ref{teo stab1}. 

Theorem \ref{teo stab1} suggests a method to build a \textit{moduli space} for the set of stabilizer states, that is, a set that parameterizes the collection of stabilizer states and endows it with some inherent nontrivial algebraic structure. Indeed, we will see (Corollary \ref{cor moduli space}) that for the collections $\mathfrak{G}$ and $\mathfrak{S}$ of stabilizer subgroups and $S$-states, respectively (Definitions \ref{def gsubgroups} and \ref{def state}) there are natural bijections
\[
\mathfrak{G}\simeq \mathfrak{S}\simeq \sqcup_{H\subset A} A\times_H{\rm Ch}_2(H)
\]
where the first identification is given by Theorem \ref{teo stab1}. The disjoint union is taken over all compact open subgroups $H\subset A$ and the set $A\times_H{\rm Ch}_2(H)$ will be suggestively constructed as ``a bundle over $A/H$ associated with the $H$-principal bundle $A\to A/ H$, with typical fiber ${\rm Ch_2}(H)$". Moreover, each fiber has a \textit{natural} structure of Abelian group. Hence, the sets $\mathfrak{G}$ and $ \mathfrak{S}$, as disjoint unions of groups, are naturally groupoids. 

%Notice the formal analogy with other classical moduli spaces, such as that of algebraic curves, which has a structure of groupoid as well; see \cite{madsen} (of course, our setting is much more elementary because the above moduli space does not seem to have an interesting topological structure). 

As a consequence of the above construction, we obtain a formula for the number of stabilizer states for an arbitrary finite Abelian group, that is  
\[
\#\mathfrak{G}=\# \mathfrak{S}=\#A\cdot \sum_{H\subset A} \#{\rm Sym}(H).
\]
Explicit formulas were previously known for certain cases, specifically for $A=\bZ^n_p$, with $p$ being a prime number or, more generally, $A=(\mathbb{F}_{p^m})^n$ ($\mathbb{F}_{p^m}$ being the Galois field) \cite{gross2006hudson} and for $A=\bZ_d^n$, with an arbitrary $d$ -- a case settled only recently in \cite{singal2023counting} (see also the references therein). 
\subsection{Wehrl entropy} 
We now consider the problem of minimizing the Wehrl entropy for a system as above. 

%Now we suppose that $\cH$, and therefore $L^2(A)$, are separable. 

For any ``reference" ket $\ket{\phi}\in \cH$, with $\|\ket{\phi}\|=1$, consider the corresponding {\it Weyl-Heisenberg coherent states} 
\begin{equation}\label{eq phiz}
\ket{\phi_z}:=W(z)\ket{\phi} \qquad z\in\AhA.
\end{equation}
We recall that, given $\ket{\phi}\in\cH$, with $\|\ket{\phi}\|=1$, and a density operator $\rho$ (a compact nonnegative operator in $\cH$ with trace $1$), the corresponding \textit{Husimi function} is given by
\begin{equation}\label{eq husimi}
u_{\phi,\rho}(z):=\langle \phi_z|\rho|\phi_z\rangle\qquad z\in\AhA.
\end{equation}

If $G:[0,1]\to\R$ is a concave function, the (generalized) Wehrl entropy of a density operator $\rho$, with respect to a reference state $\ket{\phi}\in\cH$, with $\|\ket{\phi}\|=1$, is then defined as
\begin{equation}\label{eq wehrl entropy}
\cE_G(\phi,\rho):=\int_{\AhA} G(u_{\phi,\rho}(x,\xi))\,dx\,d\xi
\end{equation}
provided the integral makes sense. 
In the following theorem we assume that $G(0)=0$, which ensures that the negative part of the integrand function is summable. Hence, the above integral is well defined, although it might be $+\infty$ (see Remark \ref{rem wehrl makessense}). 

The next result provides a full characterization of the minimal entropy states. 

\begin{theorem}\label{teo wehrl}
Let $G:[0,1]\to\R$ be a concave function with $G(0)=0$. For every $\ket{\phi}\in\cH$, with $\|\ket{\phi}\|=1$, and density operator $\rho$, we have 
\begin{equation}\label{eq 1.7}
\cE_G(\phi,\rho)\geq G(1).
\end{equation}
Moreover, if $G$ is not linear (that is, not of the form $G(\tau)=\alpha\tau$ for any $\alpha\in\R$) the following facts are equivalent.
\begin{itemize}
\item[(i)] Equality occurs in \eqref{eq 1.7} for a pair $\ket{\phi}$, $\rho$ as above.  
\item[(ii)] The Husimi function $u_{\phi,\rho}$ is the indicator function of a subset $\mathcal{K}\subset \AhA$ of measure $|\mathcal{K}|=1$.
\item[(iii)] $\ket{\phi}$ is an $S$-state (Definition \ref{def state}), $\rho=\ket{\psi}\bra{\psi}$ is a pure state, and there exist $z\in\AhA$, $\theta\in\R$ such that
\[
\ket{\psi}=e^{i\theta} W(z) \ket{\phi}.
\]
\end{itemize}
\end{theorem} 
We emphasize that this result is new even in the case of finite Abelian groups.

In Theorem \ref{teo wehrl}, instead of $S$-states we could equivalently refer to stabilizer states, in view of Theorem \ref{teo stab1}. Also, note that Theorem \ref{teo wehrl} is relevant for groups $A$ that have a compact open subgroup. Indeed, if this condition is not met, equality in \eqref{eq 1.7} never occurs. Hence, the above result is in a sense \textit{ complementary} to that for a spinless particle in $\mathbb{R}^d$ in the Schr\"odinger representation, therefore $A=\mathbb{R}^d$, originally studied by Wehrl \cite{wehrl79} in the case in which $G(\tau)=-\tau\log \tau$. There, it was conjectured that if $\ket{\phi}$ represents the ground state of the harmonic oscillator, the entropy $\cE_G(\phi,\rho)$ is minimized by the pure states $\rho=\ket{\phi_z}\bra{\phi_z}$. This conjecture was proved in \cite{lieb}, whereas the uniqueness of the minimizers was subsequently established in \cite{carlen91}, that is, the Wehrl entropy (with $\ket{\phi}$ as specified above) is minimized exclusively by these states. The analogous conjecture for the $SU(2)$ and certain $SU(N)$ coherent states was proved in \cite{lieb2014} (see also \cite{giovannetti2015}), and \cite{lieb2016}, respectively. The uniqueness of the minimizers in the $SU(2)$ case has been obtained only recently, simultaneously, and independently in \cite{frank22} and \cite{kulikov22bis}. See also \cite{frank2023generalized} for a quantitative version of this lower bound and \cite{nicola22} for closely related estimates. 

Hence, Theorem \ref{teo wehrl} solves the analog of the Wehrl conjecture for any LCA group $A$ that contains a compact open subgroup. Further results, in particular, a characterization of the maximizers of the Wehrl entropy in the finite-dimensional setting, are presented in Section \ref{sec finite dim}. 

 We are confident that the methods discussed in this note and the characterization of stabilizer states in Theorem \ref{teo stab1} can be utilized to investigate optimizers of other entropic inequalities in quantum information theory (see \cite{depalma2018} for an extensive review of current advances of this issue in the case of Gaussian channels; hence $\cH\simeq L^2(\R^d)$).  
We will postpone the examination of these topics to subsequent investigations.

\section{Tools from harmonic analysis}\label{sec 2}
\subsection{Notation}
As explained in the introduction, we denote by $A$ a locally compact Abelian (LCA) group, and by $\hA$ its dual group, whole elements are the continuous characters $\xi:A\to U(1)$. Endowed with the topology of the uniform convergence on compact subsets, $\hA$ becomes an LCA group (\cite[Theorems 23.13 and 23.15]{hewitt63}). 
We endow $A$ with {\it a} Haar measure and $\hA$ with {\it the} normalized Haar measure so that the Fourier inversion formula holds with constant $1$. We denote by $|H|$ the measure of a subset $H$ of $A$, or $\hA$ or $\AhA$ (the latter endowed with the \textit{Radon product} measure; see \cite[Section 13]{hewitt63}). 
The inner product in $L^2(A)$ is defined as
\[
\braket{\phi}{\psi}_{L^2(A)}=\int_A \overline{\phi(x)}\psi(x)\, dx,
\]
where $dx$ is the Haar measure chosen on $A$ (similarly, we write $d\xi$ for the Haar measure on $\hA$ specified above). 

We recall that the Haar measure on any LCA group assigns a positive measure to nonempty open subsets; see, e.g., \cite[Proposition 2.19]{follandbook}. Conversely, if $H\subset A$ is a compact subgroup of measure $|H|>0$ then $H$ is open. Indeed, the set where $\chi_H\ast\chi_H>0$ ($\chi_H$ denoting the indicator function of $H$) is open (because $\chi_H\ast\chi_H$ is continuous) and nonempty if $|H|>0$, and is easily seen to be contained in $H$. Hence, $H$ is open since it has a nonempty interior. 

For a subgroup $H\subset A$, the annihilator of $H$ is the closed subgroup of $\hA$ given by
\[ 
H^\bot=\{\xi\in\hA:\ \xi(x)=1\ \textrm{for all}\ x\in H\} 
\] 
(cf. \cite[Definition 23.23]{hewitt63}). We recall that if $H\subset A$ is a compact and open subgroup, the same holds for $H^\bot\subset\hA$ (\cite[Remark 4.2.22]{reiter00}). Moreover, as a consequence of the Fourier inversion formula, we have
\[
|H||H^\bot|=1
\]
(see \cite[Formula (4.4.6)]{reiter00}). We also recall that a subgroup $H\subset A$ is open if and only if $G/H$ is discrete (\cite[Theorem 5.21]{hewitt63}). 

We observe that, given a closed subgroup $H\subset A$, we have a natural isomorphim of LCA groups (\cite[Theorem 24.11]{hewitt63})
\begin{equation}\label{eq iso}
\widehat{H}\to \hA/H^\bot\qquad \xi\mapsto\underline{\xi}
\end{equation}
with $\underline{\xi}=\xi'+H^\bot$, where $\xi'\in\hA$ is any extension of $\xi\in \widehat{H}$, that is,  $\xi'|_H=\xi$.

We refer to \cite{hewitt63} for a comprehensive introduction to the theory of LCA groups. 
\begin{definition}\label{def sym}
    A continuous homomorphism $\beta:A\to\hA$ is called symmetric if $\beta(x)(y)=\beta(y)(x)$ for every $x,y\in A$. The set of continuous symmetric homorphisms $\beta:A\to\hA$ will be denoted by ${\rm Sym}(A)$. It is easily seen to be a group with respect to the product given by \[
(\beta_1\beta_2)(x):=\beta_1(x)+\beta_2(x) \qquad x\in A\]
(the sum being understood in $\hA$). 
\end{definition}
Observe that a homomorphism $\beta\in{\rm Sym}(A)$  can be identified with the continuous symmetric bicharacter $A\times A\to U(1)$ given by $\beta(x,y):=\beta(x)(y)$.

\subsection{Isotropic subgroups}
For an LCA group $A$ as above, the phase space $\AhA$ is also an LCA group. It is naturally endowed with a symplectic structure given by the bicharacter \[
\sigma:(\AhA)\times(\AhA)\to U(1)\]
defined as
\[
\sigma((x,\xi),(y,\eta))=\xi(y)\overline{\eta(x)}\qquad (x,\xi),(y,\eta)\in\AhA. 
\]
\begin{definition}
    A subgroup $K\subset\AhA$ is called {\it isotropic} if $\sigma(z,w)=0$ for every $z,w\in K$.
    \end{definition}
    The next result will be crucial later (see \cite[Propositions 3.5 and 3.6]{nicola2023maximally}). 
    \begin{proposition}\label{pro max compact open sub}
A compact open isotropic subgroup $K\subset\AhA$ necessarily has measure $|K|\leq 1$. Moreover $|K|=1$ if and only if it is a maximal compact open isotropic subgroup, that is, if it is not strictly contained in any compact open isotropic subgroup.
\end{proposition}
The class of maximal compact open isotropic subgroups plays a role analogous to that of Lagrangian subspaces in the symplectic geometry over $\R$. A class of examples is given by subgroups of product type, that is,  $K=H\times H^\bot$ for any compact open subgroup $H\subset A$. The following result from \cite[Propositions 3.1, 3.5 and 3.6]{nicola2023maximally} describes, in fact, all the maximal compact open isotropic subgroups. We will make use of the notation $\xi$ and $\underline{\xi}$ for the correspondence given in \eqref{eq iso}. Hence, if $\beta\in{\rm Sym}(H)$ and $x\in H$, we have $\beta(x)\in\widehat{H}$ and $\underline{\beta(x)} \in \hA/H^\bot$.

\begin{theorem}\label{teo iso sub}
There is a one-to-one correspondence between the class of maximal compact open isotropic subgroups $K\subset \AhA$ and the set of pairs $(H,\beta)$, where $H\subset A$ is a compact open subgroup and $\beta\in{\rm Sym}(H)$. Exactly, every such $K$ can be uniquely written in the form
\begin{align*}
K&=\{(x,\xi)\in\AhA:\ x\in H,\ \xi\in\underline{\beta(x)}\}\\
&=\{(x,\xi)\in\AhA:\ x\in H,\ \xi|_{H}=\beta(x)\}.
\end{align*}
Moreover $H^\bot= \{\xi\in\hA:\ (0,\xi)\in K\}$.
\end{theorem}
Observe that $\underline{\beta(x)}\in \hA /H^\bot$ is a subset of $\hA$ -- a coset of $H^\bot $ in $\hA$. Hence $H$ is the image of $K$ under the canonical projection $\AhA\to A$ and the fiber in $K$ over $x$ is identified with $\underline{\beta(x)}$. In algebraic terms, $K$ is regarded as an extension of $H$ by $H^\bot$, that is, we have an exact sequence of Abelian groups
\[
0\to H^\bot\to K\to H\to0,
\]
where the second arrow is the map $\xi\mapsto (0,\xi)$ and the third arrow is the restriction to $K$ of the canonical projection $(x,\xi)\mapsto x$. 

We observe that, with the above notation, if $x,y\in H$ then 
\begin{equation}
\beta(x)(y)=\xi(y)\qquad\textrm{for every}\ \xi\in\underline{\beta(x)}.
\end{equation}
Indeed, if $x,y\in H$ and $\underline{\beta(x)}=\xi'+H^\bot$, where $\xi'\in\hA$ extends $\beta(x)\in\widehat{H}$, and $\xi=\xi'+\eta$, with $\eta\in H^\bot$ we have 
\begin{equation}\label{eq 2.2bis}
\beta(x)(y)=\xi'(y)=\xi'(y)\eta(y)=(\xi'+\eta)(y)=\xi(y),
\end{equation}
where we used that $\eta(y)=1$, since $\eta\in H^\bot$ and $y\in H$.

%Similarly one defines the inner product in $\ell^2(\AhA)$, linear in the second argument. 
\subsection{Coherent state transform} Given $(x,\xi)\in\AhA$ and $\psi\in L^2(A)$, we have already defined in \eqref{eq tfshift} the phase-space shift $\pi(x,\xi):L^2(A)\to L^2(A)$.
It is easy to check that, for $x,y\in A$, $\xi,\eta\in \hA$ we have
\begin{equation}\label{eq tf product}
\pi(x,\xi)\pi(y,\eta)=\overline{\eta(x)}\pi(x+y,\xi+\eta),
\end{equation}whence the following commutation relations follow:
\begin{equation}\label{eq 2.2}
\pi(x,\xi)\pi(y,\eta)=\xi(y)\overline{\eta(x)}\pi(y,\eta)\pi(x,\xi).
\end{equation}

\begin{definition}\label{def coherent}
    Given $\phi \in L^2(A)$, the {\it coherent state transform} (alias {\it short-time Fourier transform}) of a function $        \psi\in L^2(A)$ with ``window" $\phi$ is the function
\[
V_\phi \psi(z):=\langle \pi(z)\phi,\psi\rangle_{L^2(A)}\qquad z\in\AhA.
\]
\end{definition}
This phase-space  transform is widely used in harmonic analysis and signal processing \cite{grochenig_book}, as well as in mathematical physics \cite{degosson2021quantum,lieb_book}. 
We list here some basic properties whose proof can be found (in the context of a general LCA group) in \cite{grochenig1998aspects,nicola2023maximally}. 

It is known (see, e.g., \cite[Proposition 2.1]{nicola2023maximally}) that $V_\phi \psi$ is a continuous function that vanishes at infinity, i.e., for every $\epsilon>0$ there exists a compact subset $K\subset\AhA$ such that $|V_\phi \psi(z)|<\epsilon$ for $z\in A\setminus K$.  

As a consequence of the commutation relations \eqref{eq 2.2} we have the following covariance-type properties, for $x,y\in A$, $\xi,\eta\in\hA$:  
\begin{equation}\label{eq 2.3}
V_{\phi}(\pi(x,\xi)\psi)(y,\eta)= \xi(x)\overline{\eta(x)} V_\phi \psi(y-x,\eta-\xi)
\end{equation}
and
\begin{equation}\label{eq 2.4}
V_{\pi(x,\xi)\phi}(\pi(x,\xi)\psi)(y,\eta)= \xi(y)\overline{\eta(x)}  V_\phi \psi(y,\eta).
\end{equation}
By the Cauchy-Schwarz inequality we also have
\begin{equation}\label{eq cs}
|V_\phi \psi(x,\xi)|\leq \|\psi\|_{L^2(A)} \|\phi\|_{L^2(A)}\qquad x\in A,\ \xi\in\hA.
\end{equation}
Moreover, as a consequence of the Parseval Formula for the Fourier transform on $A$, we have the following Parseval-type formula for the coherent state transform: 
\begin{equation}\label{parseval}
\int_{\AhA} |V_\phi \psi(x,\xi)|^2\, dx\,d\xi= \|\phi\|^2_{L^2(A)} \|\psi\|^2_{L^2(A)}. 
\end{equation}
The following result from \cite[Proposition 4.1]{nicola2023maximally} is less known. It provides some information on the subset of $\AhA$ where the function  $V_{\psi}\psi$ achieves its maximum value. First we observe that if $\phi\in L^2(A)$, with $\|\phi\|_{L^2(A)}=1$ then by \eqref{eq cs} we have 
\[
|V_{\phi}\phi(z)|\leq V_{\phi}\phi(0)=1\qquad z\in\AhA.
\]
\begin{proposition}\label{pro max sambiguity}
Let $\phi\in L^2(A)$, with $\|\phi\|_{L^2(A)}=1$. The set 
\[
\{z\in \AhA:\  |V_{\phi}\phi(z)|=1\}
\]
is a compact isotropic subgroup of $\AhA$. 
\end{proposition}
Finally we recall from \cite[Proposition 2.3]{nicola2023maximally} the following uniqueness result (see \cite[Section 4.2]{grochenig_book} for the case $A=\R^d$). 
\begin{proposition}\label{pro ambiguity phase}
Let $\phi,\psi\in L^2(A)$. Then 
\[
V_{\phi}\phi(z)=V_\psi \psi(z)\quad\textrm{for every}\ z\in\AhA
\]
if and only if  $\psi=e^{i\theta} \phi$ for some $\theta\in\R$. 
\end{proposition}

\subsection{Characters of second degree} The notion of {\it character of second degree} was introduced explicitly in \cite{weil64} and found many applications in harmonic analysis \cite{reiter89}, number theory \cite{igusa2012theta} and quantum information theory, where these characters are generally referred to as {\it quadratic forms}, see, e.g., \cite{nest2011}.   
\begin{definition} \label{def charac}
Consider a symmetric homomorphism $\beta\in{\rm Sym}(A)$ (Definition \ref{def sym}). 
A continuous function $h:A\to U(1)$ is a character of second degree of $A$, associated with $\beta$, if 
\[
h(x+y)=h(x)h(y)\,\beta(x)(y)\qquad x,y\in A. 
\]
The set of characters of second degree of $A$, endowed with the product defined by $(h_1h_2)(x)=h_1(x)h_2(x)$, $x\in A$, forms an Abelian group denoted by ${\rm Ch_2}(A)$. 
\end{definition}
We have the following important existence and uniqueness result.

\begin{theorem}\label{teo charac} For every symmetric homomorphism $\beta\in{\rm Sym}(A)$ there exists a corresponding character of second degree, and two characters of second degree associated with the same $\beta$ differ by the multiplication by a character of $A$. In other terms, we have the following short exact sequence of Abelian groups:
\[
0\to \widehat{A}\to {\rm Ch}_2(A)\to {\rm Sym}(A)\to0.
\]

\end{theorem} 
In fact, the existence is nontrivial (in the present generality) and was first proved in \cite[Lemma 6]{igusa68}, while the uniqueness (up to multiplication by a character) is an easy consequence of the definition, as already observed in \cite[page 146]{weil64}.  

The case of finite groups is particularly important in quantum information theory. The corresponding characters of second degree can be described explicitly as follows. 
\subsubsection*{Explicit construction for finite Abelian groups} 
Consider a finite Abelian group $A$ (endowed with the discrete topology). The following construction of characters of second degree, associated with a given symmetric homomorphism $\beta\in{\rm Sym}(A)$, is inspired by that given in \cite{reiter89} in the case where $A$ is a vector space over a local field. We use the fact that any such group is isomorphic to the direct sum of a finite number of cyclic groups. 

{\it Case $A=\mathbb{Z}_d:=\mathbb{Z}/d\mathbb{Z}$}.  We identify $\widehat{\mathbb{Z}_d}$ with $\bZ_d$ via the pairing $(x,y)\mapsto \exp(2\pi i xy/d)$. Then a (symmetric) homomorphism $\beta:\mathbb{Z}_d\to\widehat{\mathbb{Z}_d}$  is just the multiplication $x\mapsto p x$, for some $p\in \{0,\ldots,d-1\}$ (we emphasize that $p$ is here an integer). Then we consider the function 
\[
h(x)=\exp\Big(\pi i\, p\, \tilde{x}^2(d+1)/d\Big)\quad x=\tilde{x}+d\bZ\in \mathbb{Z}_d.
\]
It is easy to see that $h(x)$ is well defined, i.e. the result does not depend on the representative $\tilde{x}\in\bZ$ that we choose within its residue class $x\in \bZ_d$, and that it is indeed a character of second degree associated with $\beta$ (see e.g.\ \cite{feichtinger2,nicola2023uncertainty}).\par 
{\it Case $A=A_1\oplus...\oplus A_n$, with $A_j$ cyclic.} Each $x\in A$ can be uniquely decomposed as $x=x_1+\ldots+ x_n$, with $x_j\in A_j$. Given a symmetric homomorphism $\beta:A\to\hA$, consider the symmetric bicharacter $\beta(x,y)=\beta(x)(y)$, and let $\beta_j$ be its restriction to $A_j\times A_j$. Let $h_j$ be the corresponding character of second degree of $A_j$, as constructed above. Then it is easy to check by induction on $n$ that the function 
\[
h(x):=\prod_{j=1}^n h_j(x_j) \prod_{1\leq j<k\leq n} \beta(x_j,x_k)
\]
is a character of second degree associated with $\beta$. 

Hence, we have constructed characters of second degree associated with any given $\beta\in{\rm Sym}(A)$. By Theorem \ref{teo charac} every character of second degree is obtained from these by multiplying by a character of $A$. 

We observe that the characters of second degree of any finite Abelian group were already constructed in \cite{kaiblinger09}, relying on the Smith normal form of $A$, namely the fact that $A$ is isomorphic to the direct sum $\bZ_{d_1}\oplus \bZ_{d_1}\oplus\ldots\oplus \bZ_{d_n}$ for some $d_1,d_2,\ldots,d_n$ with $d_1|d_2|...|d_n$. The above construction offers several advantages over the one in \cite{kaiblinger09}: it does not require knowing that particular normal form, but rather just the knowledge of an isomorphism of $A$ into a direct sum of cyclic groups, and it avoids any matrix manipulation in modular arithmetic.
\begin{remark}\label{rem order}
If $x\in A$, $\xi\in\hA$, $n\in\bZ$, we have $\xi(nx)=\xi(x)^n$. Hence, if $A$ is finite, say of cardinality $N$, every character $\xi\in\hA$ takes values in the multiplicative subgroup of $U(1)$ generated by $\exp(2\pi i/N)$. Similarly, by the above construction (or by a direct inspection of the very definition, cf. [Lemma 5 (f)]\cite{nest2011}) we see that any $h\in{\rm Ch}_2(A)$ takes values in the multiplicative subgroup generated by $\zeta=\exp(\pi i/N)$. This fact is related to the definition of the Pauli group \eqref{eq pauli}; see Proposition \ref{pro struttura sottogruppo}. 
\end{remark}
\subsubsection*{Connection with the Clifford group} Consider a character $h\in{\rm Ch}_2(A)$ of second degree of $A$, associated with a symmetric homomorphism $\beta\in{\rm Sym}(A)$ (Definition \ref{def charac}). Let $C_h: L^2(A)\to L^2(A)$ be the unitary operator given by the pointwise multiplication by $h$, that is, $C_h \psi= h\psi$, for $\psi\in L^2(A)$. Then a direct inspection shows that 
\[
C_h \pi(z) C_h^\dagger= \overline{h(-x)}\pi(S z),\quad z=(x,\xi)\in\AhA,\quad  S:=\begin{pmatrix}I&0\\ \beta&I\end{pmatrix},
\]
where the matrix $S$ is regarded as a homomorphism $\AhA\to \AhA$. We notice that $S$ is a symplectic matrix, that is, $\sigma(Sz,Sw)=\sigma(z,w)$ for every $z,w\in\AhA$, because $\beta$ is symmetric. Hence $C_h$ is an element of the so-called \textit{Clifford} (alias \textit{metaplectic}) group of $A$. Indeed, it intertwines the projective unitary representations $\pi(z)$ and $\pi(Sz)$ (see \cite{reiter89,weil64}).

\subsection{Subcharacters of second degree}
The notion of subcharacter of second degree was introduced in \cite{nicola2023maximally} and will play a key role in the following.
\begin{definition}\label{def subcharac second degree}
A function $h_0:A\to\mathbb{C}$ is called a subcharacter of second degree if there exists a compact open subgroup $H\subset A$ such that $h_0(x)=0$ for $x\in A\setminus H$ and the restriction of $h_0$ to $H$ is a character of second degree of $H$. 

%If $h_0|_{H}$ is associated with the continuous symmetric homomophism $\beta\in{\rm Sym}(H)$ (Definition \ref{def charac}), we say that $h_0$ is associated with the pair $(H,\beta)$. \color{red}{SERVE?}
\end{definition}
This terminology was inspired by that of \textit{subcharacter}, that is a function $A\to\mathbb{C}$ that vanishes outside some compact open subgroup $H$ and whose restriction to $H$ is a character of $H$; see \cite[Definition 43.3]{hewitt70}. 

\begin{definition}\label{def gamma12}

Let $h_0$ be a subcharacter of second degree of $A$; hence $H:=\{x\in A:\  h_0(x)\not=0\}$ is an open compact subgroup of $A$ and $h:=h_0|_{H}\in {\rm Ch}_2(H)$. Consider the symmetric homomorphism $\beta\in{\rm Sym}(H)$ associated with $h$ according to Definition \ref{def charac}. Let $K\subset\AhA$ be the unique maximal compact open isotropic subgroup corresponding to the pair $(H,\beta)$, according to Theorem \ref{teo iso sub}. Then we will say that the subgroup $K$ is associated with the subcharacter $h_0$.  
\end{definition}

The following result from \cite[Proposition 4.4, Theorem 4.5, Theorem 5.2] {nicola2023maximally}, provides an analytic counterpart of the above algebraic constructions.   
We denote by $\chi_K$ the indicator function of a subset $K\subset\AhA$.
\begin{theorem}\label{teo nic}\ 
\begin{itemize}
\item[(i)] Let $h_0$ be a subcharacter of second degree of $A$ associated with a maximal compact open isotropic subgroup $K\subset\AhA$ in the sense of Definition \ref{def gamma12}. Then
\[
V_{h_0} h_0(x,\xi)=|H| \overline{h_0(-x)}\chi_K(x,\xi)\qquad (x,\xi)\in\AhA.
\]
Moreover, the function $|H|^{-1} V_{h_0}h_0$, restricted to $K$, is a character of second degree of $K$ associated with the symmetric homomorphism $\beta'\in{\rm Sym}(K)$ given by $\beta'(x,\xi)(y,\eta)=\overline{\eta(x)}$, that is, for $(x,\xi),(y,\eta)\in K$, 
\[ V_{h_0}h_0(x+y,\xi+\eta)= |H|^{-1} V_{h_0}h_0(x,\xi)V_{h_0}h_0(y,\eta)\overline{\eta(x)}. 
\]

\item[(ii)]
Let $\phi\in L^2(A)\setminus\{0\}$. Then $K:=\{z\in\AhA:\  V_\phi \phi(z)\not=0\}$ has measure $|K|\geq 1$. Moreover $|K|=1$ if and only if 
\[
\phi(x)=c\, h_0(x-y)
\] 
for some $c\in\mathbb{C}\setminus\{0\}$, and some subcharacter $h_0$ of second degree. In that case, $K$ is the maximal compact open isotropic subgroup associated with $h_0$ (Definition \ref{def gamma12}) and the function $\|\phi\|^{-2}_{L^2(A)}V_\phi \phi$, restricted to $K$, is a character of second degree of $K$ associated with the homomorphism $\beta'\in{\rm Sym}(K)$ given in (i). 
\item[(iii)] Let 
$\phi,\psi\in L^2(A)\setminus\{0\}$. Then $\mathcal{K}:=\{z\in\AhA:\  V_\phi \psi(z)\not=0\}$ has measure $|\mathcal{K}|\geq 1$. Moreover $|\mathcal{K}|= 1$ if and only if 
\begin{equation}
\phi(x)=c\, h_0(x-y)\qquad \psi(x)=c'\,\pi(x',\xi')\phi(x) \qquad x\in A
\end{equation}
where $c,c'\in \mathbb{C}\setminus\{0\}$, $y,x'\in A$, $\xi'\in\hA$ and $h_0$ is a subcharacter of second degree of $A$. In this case, 
$\mathcal{K}$ is a coset in $\AhA$ of the maximal compact open isotropic subgroup $K$ associated with $h_0$ (Definition \ref{def gamma12}). 
\end{itemize}
\end{theorem}
\begin{remark}\label{rem 3.5}
    In Theorem \ref{teo nic}, given a function $\phi$ in the form \eqref{eq 1.7}, we have mentioned the maximal compact isotropic subgroup $K$ associated with $h_0$ in the sense of Definition \ref{def gamma12}. Although $h_0$ is not uniquely determined by $\phi$, it follows from the results in that theorem that $K$ is, in fact, uniquely determined by $\phi$; the same remark applies to the subgroup $K$ in \eqref{eq K}. See also Lemma \ref{lemma lin ind}. 
\end{remark}
\begin{remark}
    The results in Theorem \ref{teo nic} (ii) and (iii) were formulated in  \cite[Theorems 4.5 and 5.2]{nicola2023maximally}  in a sligthly different (but equivalent) way. Precisely, instead of $\phi(x)=c\, h_0(x-y)$, there one finds that $\phi=c\, \pi(y,\eta)h_0$ for some $(y,\eta)\in \AhA$. This amounts to the same thing, because $\pi(y,\eta)h_0(x)=c'\,h'_0(x-y)$ where $c'\in\mathbb{C}\setminus\{0\}$ and $h'(x):=\eta(x)h_0(x)$ is still a subcharacter of $A$ associated (in the sense of Definition \ref{def gamma12}) with the same maximal compact open isotropic subgroup $K\subset\AhA$ as $h_0$, by Theorem \ref{teo charac}.
\end{remark}
\begin{remark}\label{rem agg0}
    The above results are presented in terms of the phase space shifts $\pi(x,\xi)$ (cf. \eqref{eq tfshift}) and the coherent state transform $V_\phi \psi$ (Definition \ref{def coherent}), hence in terms of wave functions in $L^2(A)$. One can trivially translate everything in terms of state vectors in $\cH$. Indeed, a unitary operator $U=\cH\to L^2(A)$ is given, and the Weyl-Heisenberg operators $W(x,\xi)$ are defined by $W(x,\xi)=U^\dagger \pi(x,\xi) U$, cf. \eqref{eq weylop}. Hence, for example,
    \[
\langle \psi|W(x,\xi)|\phi\rangle=\langle \psi|\pi(x,\xi)\phi\rangle_{L^2(A)} =\overline{V_\phi \psi(x,\xi)}.
    \]
    In the following, when necessary, we will freely use the above results rephrased in terms of state vectors and Weyl-Heisenberg operators. 
\end{remark}
%\begin{remark}\label{rem agg}Consider a wave function of the type $\phi(y)=c\pi(x,\xi) h_0(y)$ for some $c\in \mathbb{C}\setminus\{0\}$, $x\in A$, $\xi\in\hA$ and some subcharacter $h_0$ of second degree of $A$. Then we can write $\phi(y)= c' h'_0(y-x)$, for some $c'\in \mathbb{C}\setminus\{0\}$ (depending on $x$ and $\xi$) and with $h'(y)=\xi(y) h_0(y)$. By Theorem \ref{teo charac}, $h'$ is still a character of second degree, with $\gamma_2(h'_0)=\gamma_2(h_0)$ (with the notation introduced before Definition \ref{def gamma12}). As a consequence, the wave functions $\phi(x)$ and $\psi(x)$ that appear in Theorem \ref{teo nic} (iii) give rise to $S$-states $\ket{\phi}$ and $\ket{\psi}$ (Definition \ref{def state}). Moreover \[   \ket{\psi}=c'' W(x'-x,\xi'-\xi)\ket{\phi}. \]  and by the covariance property \eqref{} we see that  \[ V_{\phi}\phi=V_\psi \psi =\chi_K \]   where $\gamma_2(h_0)=(H,\beta)$. \end{remark}

\section{Stabilizer states}
In this section we describe the stabilizer states in terms of their wave function, as explained in the introduction. In particular, we prove Theorem \ref{teo stab1}. We begin with a number of remarks and a result that clarify the choice of the class $\mathfrak{G}$ in Definition \ref{def gsubgroups}. 

\begin{remark}\label{rem osservazioni su mainteo} \ 

    \begin{itemize}
        \item[1)] If $\mathcal{G}\subset\mathbb{H}_A$ (cf. \eqref{eq heisen}) is a subgroup such that all elements of $\mathcal{G}$ admit a common eigenvector, with eigenvalue $1$, then it is clear that the homomorphism in \eqref{eq 5g} must be injective. We will see that if, moreover, $\mathcal{G}$ is locally compact, that is, closed in $\mathbb{H}_A$, then $\mathcal{G}$ is necessarily compact (see Proposition \ref{pro no loss generality}). 
        \item[2)] The requirement, in Definition \ref{def gsubgroups}, that the image of the homomorphism \eqref{eq 5g} has measure $1$, should be interpreted as a maximality condition (see Proposition \ref{pro struttura sottogruppo}). When $A$ is finite, say of cardinality $N$, that condition is equivalent to stating that $\mathcal{G}$ has cardinality $N$. Indeed, we can consider the counting measure as the Haar measure on $A$. Then the corresponding Haar measure on $\hA$ (for which the Fourier inversion formula holds with constant $1$) will be the counting measure multiplied by $N^{-1}$.
        \item[3)] Since a compact subgroup of an LCA group has a positive measure if only if it is open, we see that $\mathfrak{G}$ is nonempty only if $A$ contains a compact open subgroup (we do not assume this condition explicitly in Theorem \ref{rem osservazioni su mainteo}, because the statement is vacuously true when it is not satisfied). Interestingly, we will show that, in a sense, there is no loss of generality in assuming that $A$ contains a compact open subgroup (see Proposition \ref{pro no loss generality}). 
        \item[4)] The following are examples of groups that contain a compact open subgroup: finite groups, or more broadly, compact or discrete groups, finite-dimensional vector spaces over the $p$-adic field $\mathbb{Q}_p$ (a case of particular interest in harmonic analysis), and products of a finite number of such groups. 
    \end{itemize}
\end{remark}
In the following, we will make use of the coherent state transform defined in Definition \ref{def coherent} (see also Remark \ref{rem agg0}), and in particular of the relationship
\begin{equation}\label{eq ambiguity}
V_\phi \phi(z)=\langle  \pi(z)\phi|\phi\rangle_{L^2(A)}=\overline{\langle \phi|W(z)|\phi\rangle}\qquad z\in \AhA.
\end{equation}
It is well known that, if $A$ is any LCA group, there exist an integer $d\geq0$ and an LCA group $A_0$, containing a compact and open subgroup, such that $A\simeq \R^d\times A_0$ (topological isomorphism); see, e.g., \cite[Theorem 24.30]{hewitt63}. 
\begin{proposition}\label{pro no loss generality} 

Let $\mathcal{G}\subset\mathbb{H}_A$ be a closed subgroup such that there exists a common eigenvector of all the elements of $\mathcal{G}$. Then $\mathcal{G}$ is Abelian and compact.

Moreover, if $A\simeq \R^d\times A_0$ (topological isomorphism), where $d\geq0$ is an integer and $A_0$ is an LCA group, then, regarding $\mathbb{H}_{A_0}$ as a subgroup of $\mathbb{H}_{A}$, we have  $\mathcal{G}\subset\mathbb{H}_{A_0}$, and the image of $\mathcal{G}$ in $A_0\times \widehat{A_0}$ under the canonical projection $\mathbb{H}_{A_0}\simeq U(1)\times A_0\times \widehat{A_0}\to A_0\times \widehat{A_0}$ is a compact isotropic subgroup $K\subset A_0\times\widehat{A_0}$ of measure $|K|\leq 1$ (where the Haar measure of $A_0\times \widehat{A_0}$ is understood). 

\end{proposition}
\begin{proof}

    Let $\ket{\psi}$ be a common normalized eigenvector of the elements of $\mathcal{G}$ (hence with eigenvalues of modulus $1$) and let $\psi(x)$ be its wave function. Let $K\subset\AhA$ be the image of $\mathcal{G}$ in $\AhA$ under the canonical projection $\mathbb{H}_A\simeq U(1)\times\AhA\to\AhA$. Since $\mathcal{G}$ is closed and $U(1)$ is compact, $K$ is closed. In addition, by \eqref{eq ambiguity},
    \[
   |V_\psi\psi(z)|=
    |\langle \pi(z)\psi|\psi\rangle_{L^2(A)} 
    |=
    |\langle \psi|W(z)|\psi\rangle |
=1\qquad z\in K.
    \]
    Hence 
    \[
    K\subset S:=\{z\in \AhA:\ |V_\psi\psi(z)|=1\}.
    \]
   By Proposition \ref{pro max sambiguity}, $S$ is a compact isotropic subgroup of $\AhA$. This implies that $K$ is compact and isotropic as well and, therefore, $\mathcal{G}$ is Abelian (cf. \eqref{eq weylop} and \eqref{eq 2.2}). In addition, $\mathcal{G}$ is compact, because, as assumed, it is a closed subset of $U(1)\times\AhA$ and is contained in the compact subset $U(1)\times K$.

   Since $\mathcal{G}$ is compact, with the identification $\mathbb{H}_A\simeq U(1)\times\AhA\simeq U(1)\times \R^d\times A_0\times \widehat{\R^d}\times \widehat{A_0}$, we see that $\mathcal{G}\subset U(1)\times \{0\}\times A_0\times \{0\}\times \widehat{A_0}\simeq\mathbb{H}_{A_0}$. That is,  $K
   \subset \{0\}\times A_0\times\{0\}\times \widehat{A_0}$.  
    
    Finally, regarded as subgroup of $A_0\times \widehat{A_0}$, $K$ is still compact and isotropic. Moreover, if $|K|>0$ then $K$ is also open in $A_0$ (as already observed, a compact subgroup of an LCA group has a positive measure if and only if it is open) and therefore, by Proposition \ref{pro max compact open sub},  $|K|\leq 1$. In all cases, we have $|K|\leq 1$. 
\end{proof}

\begin{proposition}\label{pro struttura sottogruppo}
A subset $\mathcal{G}\subset \mathbb{H}_A$ belongs to  $\mathfrak{G}$ (Definition \ref{def gsubgroups}) if and only if   
\begin{equation}\label{eq aw}
\mathcal{G}=\{\alpha(z)W(z):\ z\in K\},
\end{equation} where

\begin{itemize} \item[(i)] $K\subset\AhA$ is a maximal compact open isotropic subgroup;  
\item[(ii)] $\alpha\in{\rm Ch}_2(K)$ is associated (in the sense of Definition \ref{def charac}) with the symmetric homomorphism $\beta'\in{\rm Sym}(K)$ given by 
\[
\beta'(x,\xi)(y,\eta)=\overline{\eta(x)} \qquad (x,\xi),\,(y,\eta)\in K,
\]
that is, $\alpha:K\to U(1)$ is continuous and satisfies  
\begin{equation}\label{eq cond alpha}
\alpha(x+y,\xi+\eta)=\alpha(x,\xi)\alpha(y,\eta)\overline{\eta(x)}\qquad (x,\xi),\,(y,\eta)\in K.
\end{equation}
\end{itemize} 
 Moreover, if $A$ is finite, every subgroup $\mathcal{G}\in\mathfrak{G}$ is in fact a subgroup of the Pauli group $\mathbb{P}_A$ (cf. \eqref{eq pauli}). 
\end{proposition}
\begin{proof}
Let $\mathcal{G}$ be as in the statement. Hence $\mathcal{G}$ has the form in \eqref{eq aw} for some maximal compact open isotropic subgroup $K\subset\AhA$ and some $\alpha\in{\rm Ch}_2(K)$ satisfying \eqref{eq cond alpha}. It is clear that the map \eqref{eq 5g} is injective. Moreover, since $K\subset\AhA$ is a subgroup and the function $\alpha:K\to U(1)$ satisfies \eqref{eq cond alpha}, using \eqref{eq weylop} and \eqref{eq tf product} we see that $\mathcal{G}$ is a subgroup of $\mathbb{H}_A$. 
 Also, recall that $\mathbb{H}_A$ is identified with $U(1)\times\AhA$, as a topological space, and $\mathcal{G}$ in \eqref{eq aw} is consequently  identified with the graph of the function $\alpha$, that is, $\{(\alpha(z),z):\ z\in K\}\subset U(1)\times \AhA$. Since $U(1)$ is Hausdorff and $\alpha$ is continuous by assumption, its graph is closed in $U(1)\times K$, and therefore is compact. 
To conclude that $\mathcal{G}\in\mathfrak{G}$ it remains to check that $|K|=1$. But this follows from the characterization of the maximal compact open isotropic subgroups of $\AhA$ in Proposition \ref{pro max compact open sub}. 

%Indeed, if there were an Abelian compact subgroup $\mathcal{G}\subset\mathbb{H}_A$, not containing multiples of the identity other than the identity itself, such that its projection on $\AhA$ is a compact open subgroup, such a projection would be a compact open isotropic subgroup containing strictly $K$, which contradicts the maximality of $K$ (here we use that there are no two elements in $\mathcal{G'}$ of the type $c_1 W(z)$ and $c_2 W(z)$ with $c_1\not=c_2$). 

Conversely, let $\mathcal{G}\in\mathfrak{G}$. Since the homomorphism $\eqref{eq 5g}$ is injective, and $\mathcal{G}$ is compact, we see that $\mathcal{G}$ has the form \eqref{eq aw} for some function $\alpha:K\to U(1)$, where the subgroup $K\subset\AhA$ is compact. Moreover, since $\mathcal{G}$ is Abelian, by \eqref{eq weylop} and \eqref{eq 2.2}  we see that $K$ is isotropic. Also, $K$ has measure $|K|=1$ by assumption, and therefore is open in $\AhA$. As a consequence of Proposition \ref{pro max compact open sub}, it is a maximal compact open isotropic subgroup of $\AhA$. Moreover, the function $\alpha$ satisfies \eqref{eq cond alpha} because $\mathcal{G}$ is a subgroup of $\mathbb{H}_A$. It remains to prove that $\alpha$ is continuous. As already observed, its graph is identified with the subgroup $\mathcal{G}$, which is compact by assumption, and therefore it is closed as a subset of $U(1)\times K$. Since $U(1)$ is compact and Hausdorff, we deduce that $\alpha$ is continuous, by the closed graph theorem in point-set topology.  

%Finally, to check the maximality of $K$, we argue by contradiction. Suppose that $K'$ is an isotropic compact open subgroup of $\AhA$ containing strictly $K$. By Theorem \ref{} there exists $\alpha'\in {\rm Ch_2}(K')$ that extends $\alpha$, and satisfies the same property \eqref{} for $(x,\xi),(y,\eta)\in K'$. By the first part of this proof, the pair $(K',\alpha')$ gives rise to a compact Abelian subgroup $\mathcal{G}'\subset\mathbb{H}_A$, not counting multiples of the identity other than the identity itself. Since $\mathcal{G}'$ contains $\mathcal{G}$ strictly, this contradicts the maximality of $K$. 

Finally, let $A$ be a finite Abelian group, say of cardinality $N$. Since $|K|=1$, from Remark \ref{rem osservazioni su mainteo} 2) we see that $K$ has cardinality $N$. Hence, by Remark \ref{rem order} we have that $\alpha\in{\rm Ch}_2(K)$ takes values in the multiplicative group generated by $\zeta=\exp(\pi i/N)$. Therefore, $\mathcal{G}\subset\mathbb{P}_A$.
\end{proof}

%Notice that $\overline{V_\phi \phi(z)}=\Tr\ket{\phi}\bra{\phi}W(z)$ is the so-called characteristic function of the corresponding vector $\ket{\phi}\in\cH$ (\cite[Section 12.3]{holevo2019quantum}). 
\begin{proposition}\label{pro 3g}
    Let $\mathcal{G}\in\mathfrak{G}$ be as in Proposition \ref{pro struttura sottogruppo}, that is, in the form \eqref{eq aw} where $K\subset\AhA$ is a maximal compact open isotropic subgroup and $\alpha\in{\rm Ch_2}(K)$ satisfies \eqref{eq cond alpha}.
    
    A vector $\ket{\phi}\in\cH$, with $\|\ket{\phi}\|=1$, is a common eigenvector, with eigenvalue $1$, of all elements of $\mathcal{G}$ if and only if 
    \begin{equation}\label{eq 3g 1}
    V_{\phi} \phi(z)=
\begin{cases}
\alpha(z)&z\in K\\
0 & z\not\in K.
\end{cases}
    \end{equation}
    \begin{proof}
        Suppose that $\phi\in L^2(A)$, with $\|\phi\|_{L^2(A)}=1$,  satisfies \eqref{eq 3g 1}. Then, by \eqref{eq ambiguity}, we have 
        \[
        \langle \phi|W(z)|\phi\rangle=\overline{\alpha(z)}\qquad z\in K. 
        \]
        Since $|\alpha(z)|=1$, this means that, if $z\in K$, $\phi$ and $W(z)\phi$ are linearly dependent and that, in fact, \[
        \alpha(z) W(z)\ket{\phi}=\ket{\phi}\qquad z\in K.\]
        Conversely, it is clear that the latter formula implies that  $V_\phi \phi (z)=\alpha(z)$ for $z\in K$. Moreover $|K|=1$ by Proposition \ref{pro max compact open sub}. 
        Hence, $V_\phi \phi(z)=0$ for almost every $z\in \AhA\setminus K$, by \eqref{parseval}. Since $\AhA\setminus K$ is open and $V_\phi \phi$ is continuous, we see that $V_\phi \phi(z)=0$ for every $z\in \AhA\setminus K$. Hence \eqref{eq 3g 1} is satisfied. 
    \end{proof}

\end{proposition}
We can now prove Theorem \ref{teo stab1}.
\begin{proof}[Proof of Theorem \ref{teo stab1}]
We use the description of the subgroups $\mathcal{G}\in\mathfrak{G}$ given in Proposition \ref{pro struttura sottogruppo}. It is clear from Proposition \ref{pro 3g} that every $\ket{\phi}\in\cH$, with $\|\ket{\phi}\|=1$, is stabilized by \textit{at most} one subgroup $\mathcal{G}\in\mathfrak{G}$, because the relation \eqref{eq 3g 1} determines uniquely $K$ and $\alpha$, and therefore $\mathcal{G}$, from $\phi$. On the other hand we see that every $\mathcal{G}$ stabilizes \textit{at most} one state, again by the relationship \eqref{eq 3g 1} and the uniqueness result in Proposition \ref{pro ambiguity phase}. 

Now, let $\ket{\phi}$ be an $S$-state vector, with $\|\ket{\phi}\|=1$. Hence
\[
\phi(x)=c\, h_0(x-y)\]
for some $c\in\mathbb{C}\setminus\{0\}$, $y\in A$ and some subcharacter of second degree $h_0$. Since $\phi$ is normalized, with the notation in the statement, $|c|=|H|^{-1/2}$.  Let us now prove that $\ket{\phi}$ is stabilized by the subgroup $\mathcal{G}\in\mathfrak{G}$ given in the statement. We have to prove that the pair $K,\alpha$ given in \eqref{eq K} and \eqref{eq alpha} satisfies the conditions in Proposition \ref{pro struttura sottogruppo} and that \eqref{eq 3g 1} is satisfied. It is clear from Theorem \ref{teo iso sub} that the set $K$ in \eqref{eq K} is a maximal compact open isotropic subgroup of $\AhA$. In fact, it is the subgroup associated with the subcharacter $h_0$ in the sense of Definition \ref{def gamma12}. Moreover, by the covariance property \eqref{eq 2.4} and Theorem \ref{teo nic} (i) we have 
\begin{equation}\label{eq 3g 2}
V_\phi \phi(x,\xi)=|H|^{-1}\overline{\xi(y)}V_{h_0}h_0(x,\xi)= \overline{h_0(-x)
\xi(y)}\chi_K(x,\xi). 
\end{equation}
This implies that \eqref{eq 3g 1} is satisfies for the function $\alpha$ defined in \eqref{eq alpha}. Moreover, Theorem \ref{teo nic} (ii) tells us that the function $V_\phi \phi$, restricted to $K$, that is $\alpha$,  is a character of second degree of $K$ associated with the homomorphism $\beta'\in{\rm Sym}(K)$ given in Proposition \ref{pro struttura sottogruppo}. This concludes the proof of the first part of the statement. 

Now, consider a subgroup $\mathcal{G}\in\mathfrak{G}$. We know from Proposition \ref{pro struttura sottogruppo} that $\mathcal{G}$ has the form in \eqref{eq aw}, where $K\subset\AhA$ is a maximal compact open isotropic subgroup and $\alpha\in{\rm Ch}_2(K)$ satisfies \eqref{eq cond alpha}. Moreover $K$ has necessarily the form given in Theorem \ref{teo iso sub} for some compact open subgroup $H\subset A$ (the image of $K$ under the canonical projection $\AhA\to A$) and some symmetric homomorphism $\beta\in{\rm Sym}(H)$. Let us prove that there exists an $S$-state that is stabilized by $\mathcal{G}$.  The proof is constructive and follows steps (i)---(iii) in the statement.

(i)  We know from Theorem \ref{teo iso sub} that $H^\bot=\{\xi\in\hA:\ (0,\xi)\in K\}$.   For $\xi,\eta\in H^\bot$, by \eqref{eq cond alpha} we have 
\[
\alpha(0,\xi+\eta)=\alpha(0,\xi)\alpha(0,\eta)\eta(0)=\alpha(0,\xi)\alpha(0,\eta).
\]
Hence the map $H^\bot\ni \xi\mapsto \alpha(0,\xi)$ is a continuous character of $H^\bot$.  
The continuous character $H^\bot\ni \xi\mapsto \overline{\alpha(0,\xi)}$ admits an extension to a continuous character of $\hA$ (see, e.g., \cite[Corollary 24.12]{hewitt63}), that is an element of the dual group of  $\hA$, which is canonically identified with $A$ (see, e.g., \cite[Theorem 24.8]{hewitt63}). Therefore there exists $y\in A$ such that $\xi(y)= \overline{\alpha(0,\xi)}$ for every $\xi\in H^\bot$. 

(ii) First, let us prove that the function $h:H\to U(1)$ in the statement is well defined: if $(x,\xi), (x,\eta)\in K$, then $\eta':=\eta-\xi\in H^\bot$ by Theorem \ref{teo iso sub}. Hence by \eqref{eq cond alpha}, 
\begin{align*}
\alpha(x,\eta)\eta(y)=\alpha(x,\xi+\eta')(\xi+\eta')(y)&=\alpha(x,\xi)\alpha(0,\eta')\overline{\eta'(x)}\xi(y)\eta'(y)\\
&= \alpha(x,\xi)\xi(y),
\end{align*}
because $\eta'(x)=1$ ($x\in H$ and $\eta'\in H^\bot)$ and $\eta'(y)=\overline{\alpha(0,\eta')}$ by the point (i). 

Let us prove that $h\in {\rm Ch}_2(H)$, and precisely that $h$ is associated (in the sense of Definition \ref{def charac}) with the aforementioned symmetric homomorphism $\beta$. We have $h(x)=\overline{\alpha(-x,-\xi)}\xi(y)$ for $(x,\xi)\in K$. Hence   by \eqref{eq cond alpha}, if $(x,\xi),(x',\xi')\in K$,
\begin{align*}
h(x+x')&=\overline{\alpha(-x-x',-\xi-\xi')}(\xi+\xi')(y)\\
&=\overline{\alpha(-x,-\xi)\alpha(-x',-\xi')}\xi'(x)\xi(y)\xi'(y)\\
&=h(x)h(x')\xi'(x)\\
&=h(x)h(x')\beta(x)(x'),
\end{align*}
The last equality is true because $\beta(x)(x')=\beta(x')x=\xi'(x)$ by \eqref{eq 2.2bis}. 

It remains to prove that $h$ is continuous. It suffices to prove that there exists a continuous section of $K$, that is, a continuous map $\gamma:H\to\hA$ such that $(x,\gamma(x))\in K$ for every $x\in H$. The continuity of $h$ then follows, because $h(-x)=\overline{\alpha(x,\gamma(x))\gamma(x)(y)}$ and $\alpha:K\to U(1)$ is continuous. To construct such a section, observe that, since $H^\bot$ is open, $\widehat{H}\simeq\hA/H^\bot$ is discrete, and $\beta(H)\subset \widehat{H}$ is compact, therefore finite. Let $\beta(H)=\{\eta_1,\ldots,\eta_N\}$, and set $H_j:=\beta^{-1}(\{\eta_j\})\subset H$, for $j=1,\ldots,N$. Then each $H_j$ is open and the collection of the $H_j$'s is a partition of $H$. Now, for each $j=1,\ldots,N$, pick an element $\xi_j\in\hA$ such that $\xi_j|_{H}=\eta_j$ (such an extension always exists, see \cite[Corollary 24.12]{hewitt63}). Define $\gamma:H\to\hA$ by setting $\gamma(x)=\xi_j$ if $x\in H_j$. It is clear that $\gamma$ is continuous and that $(x,\gamma(x))\in K$ for every $x\in H$ by construction (recall that $K$ and $\beta$ are related as in Theorem \ref{teo iso sub}).  

(iii) We now extend $h$ to a function $h_0:A\to\mathbb{C}$ by setting $h_0(x)=h(x)$ for $x\in H$ and $h_0(x)=0$ for $x\in A\setminus H$. Hence $h_0$ is a subcharacter of second degree associated with $K$ in the sense of Definition \ref{def gamma12}. For the (normalized) function $\phi(x)=|H|^{-1/2} h_0(x-y)$ we see, arguing as above, that \eqref{eq 3g 2} holds, and therefore $\eqref{eq 3g 1}$ is satisfied because $\overline{h(-x)\xi(y)}=\alpha(x,\xi)$ for $(x,\xi)\in K$, by the very definition of $h$. By Proposition \ref{pro 3g}, this implies that $\ket{\phi}$ is stabilized by the group $\mathcal{G}\in\mathfrak{G}$. This concludes the proof. 
\end{proof}
\begin{remark}
As a consequence of Theorem \ref{teo stab1}, the support of the wave function $\phi$ of a stabilizer state is a coset of a compact open subgroup $H\subset A$ and $\phi$ has constant modulus on its support. When $A$ is finite these properties also follow from general results concerning \textit{monomial stabilizer states}; see  \cite{nest2011monomial,nest2011}. 
\end{remark}

\section{A moduli space for the collection of stabilizer states}
In this section we provide a natural parametrization of the set of stabilizer states (Definition \ref{def gsubgroups}). We know from Theorem \ref{teo stab1} that the stabilizer states are precisely the $S$-states (Definition \ref{def state}). Their wave functions have therefore the form in \eqref{eq stab phi} for some $y\in A$ and some subcharacter of second degree $h_0$ of $A$. However, different pars $y,h_0$ might give rise to the same state. This fact is clarified by the following result. 
\begin{lemma}\label{lemma lin ind}
Let 
\begin{equation}\label{eq phi1}
\phi(x)=c\,  h_0(x-y)\qquad x\in A
\end{equation}
and 
\begin{equation}\label{eq phi2}
\phi'(x)=c\, h'_0(x-y')\qquad x\in A,
\end{equation}
where $c,c'\in\mathbb{C}\setminus\{0\}$, $y,y'\in A$ and $h_0, h'_0$ are subcharacters of second degree of $A$, hence $H:=\{x\in A:\ h_0(x)\not=0\}$ and $H':=\{x\in A:\ h'_0(x)\not=0\}$ are compact open subgroups of $A$, and  $h:={h_0}|_{H}\in {\rm Ch}_2(H)$, $h':={h'_0}|_{H'}\in {\rm Ch}_2(H')$. Suppose that $h$ is associated with $\beta\in{\rm Sym}(H)$ and $h'$ is associated with $\beta'\in{\rm Sym}(H')$ in the sense of Definition \ref{def charac}. 

 Then $\phi(x)$ and $\phi'(x)$ are linearly dependent if and only if 
\begin{equation}\label{eq equiv}
\begin{cases}
    H=H'\\
    y-y'\in H\\
h(x)=h'(x)\beta'(y-y')(x)\qquad x\in H.
\end{cases}
\end{equation}
In this case, we also have $\beta=\beta'$. 
\end{lemma}
\begin{proof}
    Suppose that $\phi$ and $\phi'$ are linearly dependent. Then $\{x\in A:\ \phi(x)\not=0\}=\{x\in A:\ \phi'(x)\not=0\}$, which implies $H=H'$.
    
    Moreover, for some $c''\in\mathbb{C}\setminus\{0\}$,
\[ 
h_0(x-y)=c'' h'_0(x-y')\quad x\in A,   
\]
that is, 
\[ 
h_0 (x) = c'' h'_0 (x + y-y')\qquad x\in A.
\]
This implies that $y-y'\in H$ and that 
    \[
    h(x)=c'' h'(x+y-y')\qquad x\in H,
    \]
    whence
    \[
    h(x)=c'' h'(x)h'(y-y')\beta'(y-y')(x).
    \]
    Setting $x=0$ we see that $c''h'(y-y')=1$ and therefore 
    \[
    h(x)=h'(x)\beta'(y-y')(x)\qquad x\in H.
    \]
    Since $h$ and $h'$ differs by the multiplication by a character, by Theorem \ref{teo charac} we have $\beta=\beta'$.

Conversely it is clear, by reversing the above reasoning, that condition \eqref{eq equiv} implies that $\phi$ and $\phi'$ are linearly dependent. 
\end{proof}
We now fix a compact open subgroup $H\subset A$. Using the notation of Lemma \ref{lemma lin ind}, and inspired by that result, we give the following definition. 
\begin{definition}\label{def associate bundle}
    In the set $A\times{\rm Ch_2}(H)$ we say that a pair $(y,h)$ is equivalent to $(y',h')$, and we write $(y,h)\sim (y',h')$, if $y-y'\in H$ and $h(x)=h'(x)\beta'(y-y')(x)$ for $x\in H$. We denote by 
    \[
A\times_H{\rm Ch}_2(H):= A\times{\rm Ch_2}(H)/\sim
    \]
    the corresponding quotient set. 
\end{definition}
Indeed, it is easy to see that the above is an equivalence relation. 
\begin{remark}\label{rem suggestive} The above construction admits a suggestive interpretation. That is, $A\times_H{\rm Ch}_2(H)$ can be regarded as ``a bundle over $A/H$, associated with the $H$-principal bundle $A\to A/H$, and with typical fiber ${\rm Ch}_2(H)$". Indeed, we know that $H$ acts naturally on $A$ via the map $y\mapsto u+y$, $u\in H$, $y\in A$, and this action is free and transitive on the fibers of the map $A\to A/H$. Moreover, $H$ also acts on ${\rm Ch}_2(H)$ by $u\cdot h= h_u $, with $h_u(x)=h(x)\beta(u)(x)$, for $u\in H$, $x\in A$. We now consider the diagonal action on $A\times{\rm Ch}_2(H)$, given by 
\[
u\cdot(y,h)=(y+u,h_u)\qquad u\in H,\ y\in A,\ h\in{\rm Ch}_2(H).
\]
Then $A\times_H{\rm Ch}_2(H)$ is just the set of the orbits of this action. 
\end{remark}
As a consequence of Lemma \ref{lemma lin ind} we have the following result. 
\begin{proposition}\label{pro param}
    Let $H\subset A$ be a compact open subgroup. There is a natural one-to-one correspondence between the set $A\times_H{\rm Ch}_2(H)$ and that of $S$-states (Definition \ref{def state}) whose wave functions have a coset of $H$ as support. Exactly, with $[(y,h)]\in A\times_H{\rm Ch}_2(H)$ it is associated the $S$-state represented by the wave function $\phi(x)=h_0(x-y)$, where $h_0(x)=h(x)$ for $x\in H$ and $h_0(x)=0$ for $x\in A\setminus H$. 
\end{proposition}
We now consider the natural surjective map
\begin{equation}\label{eq fibration}
A\times_H{\rm Ch}_2(H)\to A/H \qquad [(y,h)]\mapsto [y],
\end{equation}
where the square brackets represent the residue classes in the corresponding quotient spaces. It is easy to see that this map is well defined, and we now consider its fibers. Given $h,h'\in {\rm Ch}_2(H)$ we denote by $\beta$ and $\beta'\in{\rm Sym}(H)$ the corresponding symmetric homomorphisms, as in Definition \ref{def charac}. 

\begin{proposition}\label{pro group fiber} 
    The fibers of the map \eqref{eq fibration} have a natural structure of Abelian group, where the group law is defined as follows. If $[(y,h)], [(y',h')]\in A\times_H{\rm Ch}_2(H)$, with $y-y'\in H$,
    \[
[(y,h)]\cdot [(y',h')]:=[(y,h'')]
    \]
    where 
    \[
h''(x)=h(x)h'(x)\beta'(y-y')(x).
   \]
   Furthermore, each fiber is isomorphic to ${\rm Ch}_2(H)$, although not canonically. 
\end{proposition}
\begin{proof}
    It is easy to see, by direct inspection, that the above product is well defined and endows each fiber with a structure  of Abelian  group. One has to use the fact that, with the notation in the statement, $hh'\in{\rm Ch_2}(H)$ is associated with the symmetric homomorphism $\beta\beta'\in{\rm Sym}(H)$, in the sense of Definition \ref{def charac}, where $\beta\beta'(x):=\beta(x)+\beta'(x)$ for $x\in H$ (the sum being understood in $\widehat{H}$). 
    
    Concerning the second part of the statement, we observe that, for any given $y\in A$, the map 
    \[
    {\rm Ch_2}(H)\to A\times_H{\rm Ch_2(A)}\qquad h\mapsto [(y,h)]
    \]
    is a group isomorphism onto its image, and that the latter is the fiber of the map \eqref{eq fibration} over $[y]$. Injectivity follows from the fact that the action of $H$ on the coset $y+H$ (that is the fiber of the map $A\to A/H$ over $[y]$) is free, while surjectivity is a consequence of that action being transitive. 
\end{proof}
    
\begin{remark} The group law introduced in Proposition \ref{pro group fiber} on the fibers of the map \eqref{eq fibration} can be interpreted in terms of pointwise multiplication of wave functions. Namely, according to Proposition \ref{pro param}, to the pairs $[(y,h)]$ and $[(y',h')]$ in the same fiber (hence $y-y'\in H$) there correspond two wave functions $\phi$ and $\phi'$ (defined up to nonzero multiplicative constant) with the same support -- a coset of $H$ in $A$. Then, with the notation of Lemma \ref{lemma lin ind} we have, for some constant $c\not=0$ and every $x\in A$, 
\begin{align*}
    \phi(x)\phi'(x)&=c\, h_0(x-y) h'_0(x-y')\\
    &=c\, h_0(x-y)h'_0(x-y)\beta'(y-y')(x)\\
    &=C(y,y')h_0(x-y) h'_0(x-y)\beta'(y-y')(x-y)
\end{align*}
for some number $C(y,y')\in \mathbb{C}\setminus\{0\}$ depending on $y$ and $y'$. The last expression is the wave function of an $S$-state associated, in the sense of Proposition \ref{pro param}, with $[(y,h)]\cdot [(y',h')]$ (in fact, one could also interpret this product in terms of composition of shift and Clifford operators; cf. Section \ref{sec 2}).
\end{remark} 
We can summarize the above results and Theorem \ref{teo stab1} as follows.

\begin{corollary}\label{cor moduli space}
Let $\mathfrak{G}$ and $\mathfrak{S}$ be the collections of stabilizer subgroups $\mathcal{G}\subset\mathbb{H}_A$ and $S$-states, respectively (see Definitions \ref{def gsubgroups} and \ref{def state}). Then we have natural bijections
\[
\mathfrak{G}\simeq \mathfrak{S}\simeq \sqcup_{H\subset A} A\times_H{\rm Ch}_2(H)
\]
where the first correspondence is given in Theorem \ref{teo stab1}, while the second is given in  Proposition \ref{pro param}, and the disjoint union is taken over all compact open subgroups $H\subset A$. 

Moreover, if A is finite, the cardinality of these sets is given by 
\begin{equation}\label{eq card}
\#\mathfrak{G}=\# \mathfrak{S}=\#A\cdot \sum_{H\subset A} \#{\rm Sym}(H).
\end{equation}
\end{corollary}

 \begin{proof}
It remains only to prove the last part of the statement, when $A$ is finite. We clearly have
\[
\#\mathfrak{G}=\# \mathfrak{S}=\sum_{H\subset A}\#A\times_H{\rm Ch_2}(H).
\]
On the other hand, by Proposition \ref{pro group fiber},
\[
\#A\times_H{\rm Ch_2}(H)=\# A/H\cdot \#{\rm Ch}_2(H)
\]
and by Theorem \ref{teo charac} we have 
\[
\#{\rm Ch}_2(H)=\#\widehat{H}\cdot \#{\rm Sym}(H).
\]
Since $\#A/H=\#A/\#H=\#{A}/\#\widehat{H}$, we obtain the desired conclusion. 
 \end{proof}
 \section{Minimizing the Wehrl entropy}
 In this section we prove Theorem \ref{teo wehrl}. 
 We preliminarily observe that, by the spectral theorem,  a density operator $\rho$ in $\cH$ (that is, a compact nonnegative operator with trace $1$) can always be written as 
 \begin{equation}
     \label{eq rho decomp}
     \rho=\sum_j p_j\ket{\psi_j}\bra{\psi_j},
 \end{equation}
 where $p_j>0$, with $\sum_j p_j=1$, and the $\ket{\psi_j}$'s are an \textit{orthonormal system} of $\cH$. The set of the $j$'s is at most countable, even if $\cH$ is not separable. Indeed, we have the orthogonal decomposition in $\rho$-invariant subspaces $\cH={\rm Ker}\,\rho\oplus \overline{\rho(\cH)}$, and $\overline{\rho(\cH)}$ is separable, since $\rho$ is compact.
 
As a consequence of \eqref{eq rho decomp}, for the corresponding Husimi function in \eqref{eq husimi} we have
\begin{equation}
    \label{eq husimi decomp}
   u_{\phi,\rho}(z)=\sum_{j}p_j|V_\phi \psi_j|^2,
\end{equation}
for any $\ket{\phi}\in\cH$, with $\|\ket{\phi}\|=1$, where $V_\phi \psi_j$ denotes the coherent state transform of $\psi_j$ with window $\phi$ (see Definition \ref{def coherent}).
\begin{lemma}\label{lem husimi}
Let $\ket{\phi}\in\cH$, with $\|\ket{\phi}\|=1$ and let $\rho$ be a density operator in $\cH$. The Husimi function $u_{\phi,\rho}$ satisfies 
\begin{equation}\label{eq husimi bounds1}
0\leq u_{\phi,\rho}(z)\leq 1\qquad z\in\AhA
\end{equation}
and 
\begin{equation}\label{eq husimi bounds2}
\int_{\AhA} u_{\phi,\rho}(x,\xi)\, dx\, d\xi=1.
\end{equation}
Moreover $u_{\phi,\rho}$ is a continuous function in $\AhA$, which tends to zero at infinity. In particular, the set $\{z\in\AhA:\ u_{\phi,\rho}(z)>0\}$ is $\sigma$-compact.   
\end{lemma}
\begin{proof}
    We use the formula \eqref{eq husimi decomp}. Then \eqref{eq husimi bounds1} and \eqref{eq husimi bounds2} are  immediate consequence of \eqref{eq cs} and \eqref{parseval}. 
Moreover, we know that the functions $V_\phi \psi_j$ are continuous in $\AhA$ and tend to zero at infinity (see Section \ref{sec 2}). The same holds true for $u_{\phi,\rho}$, because the series (possibly a finite sum) \eqref{eq husimi decomp} converges uniformly in $\AhA$ by the Weierstrass $M$-test. 
 \end{proof}
   The following remark clarifies the hypothesis ``$G(0)=0$" in Theorem \ref{teo wehrl}.
 \begin{remark}\label{rem wehrl makessense}
 \ 
 \begin{itemize}
     \item[1)] The assumption  ``$G(0)=0$" in Theorem \ref{teo wehrl}  ensures that the integral in \eqref{eq wehrl entropy} is well defined. Indeed, this property implies that $G(\tau)\geq \tau G(1)$ for $0\leq \tau\leq 1$; in particular $G\geq0$ if $G(1)\geq 0$. Instead, if $G(1)<0$ then $G(u_{\rho,\phi})_{-}$, the negative part  of $G(u_{\rho,\phi})$, satisfies 
     \[
0\leq G(u_{\rho,\phi})_{-}\leq |G(1)|u_{\rho,\phi},
     \]
     and is therefore summable by \eqref{eq husimi bounds2}. 
      \item[2)]
     The following example shows that for certain LCA groups $A$ and concave functions $G$, with $G(0)\not=0$, the integral in \eqref{eq wehrl entropy} is not well defined.
     
     Let $\cH=L^2(A)$, with $A=U(1)$, equipped with the normalized Haar measure. Then $\hA\simeq\bZ$, endowed with the counting measure as Haar measure. Let $\phi(x)= 1$ and  $\psi_j(x)=x^{j}$, for $x\in U(1)$, $j=1,2,\ldots$. An explicit computation shows that $|V_{\phi}\psi_j(x,\xi)|$ is the indicator function of the set $U(1)\times \{j\}\subset U(1)\times\bZ$. Setting
     \[
     \rho=\sum_{j=1}^\infty 2^{-j}\ket{\psi_j}\bra{\psi_j},
     \]
   we see that $u_{\phi,\rho}=\sum_{j=1}^\infty 2^{-j}|V_\phi \psi_j|^2$ vanishes on a subset of $U(1)\times \bZ$ of infinite measure, and for every $\epsilon>0$, we have $0<u_{\phi,\rho}<\epsilon$ on a subset of infinite measure. 
     Consider therefore any concave function $G:[0,1]\to \R$ with $G(0)<0$ and $\lim_{\tau\to 0^+}G(\tau)>0$. Then both the positive part and the negative part of the function $G(u_{\phi,\rho}(z))$ have an infinite integral on $U(1)\times\bZ$. Hence, the integral in \eqref{eq wehrl entropy} would not make sense in this case. 
     \item[3)] To ensure that the integral in \eqref{eq wehrl entropy} is well defined, we could even consider concave functions $G:[0,1]\to\R$ continuous at $0$ (without assuming that $G(0)=0$). However, when $|A|=\infty$, that integral would always be $+\infty$ if $G(0)>0$ or always $-\infty$ if $G(0)<0$, due to Lemma \ref{lem husimi}. 
     \end{itemize}
 \end{remark}
 We can now prove Theorem \ref{teo wehrl}.
 \begin{proof}[Proof of Theorem \ref{teo wehrl}]
Without loss of generality we can assume that $G(1)=0$. Indeed, if the theorem has been proved in this case, given a function $G$ as in the statement, we can then apply the theorem to the function $G(\tau)-\tau G(1)$ and use \eqref{eq husimi bounds2}.
 
 Hence, thereafter we assume that $G(1)=0$. The formula \eqref{eq 1.7} is then clear, because $G$ is nonnegative (recall that $G$ is concave and $G(0)=0$ by assumption). 
 
 Concerning the remaining part of the statement, we observe that the assumption that $G$ is not linear means (for our function $G$ satisfying $G(1)=0$) that $G(\tau)>0$ for $0<\tau<1$.  Under this assumption, we now prove the following chain of implications between the points in the statement. 
 
$(i)\Longrightarrow (ii)$
Since $G(0)=G(1)=0$ and $G(\tau)>0$ for $0<\tau<1$ we see that if equality occurs in \eqref{eq 1.7} for some $\phi$ and $\rho$ then the subset of $\AhA$ where $0<u_{\phi,\rho}<1$ has measure zero.  But this subset is open, because  $u_{\phi,\rho}$ is continuous (Lemma \ref{lem husimi}), and therefore it is empty. It follows that $u_{\phi,\rho}$ is the indicator function of some subset of $\mathcal{K}\subset \AhA$, necessarily having  measure $|\mathcal{K}|=1$, by \eqref{eq husimi bounds2}. 
 
$(ii)\Longrightarrow (iii)$ Using the spectral decomposition \eqref{eq rho decomp} for $\rho$, and therefore \eqref{eq husimi decomp}, we see that 
\[
\sum_j p_j |V_\phi \psi_j|^2=\chi_\mathcal{K},
\]
where $|\mathcal{K}|=1$. Since $p_j>0$ for all $j$, and $\sum_j p_j=1$, we have
\begin{equation}\label{eq vphipsik}
|V_\phi \psi_j|=\chi_\mathcal{K}\quad\textrm{for every}\ j.
\end{equation}
 By Theorem \ref{teo nic} (iii)  and (ii) we deduce that $\ket{\phi}$ is an $S$-state,
\begin{equation}\label{eq 24g}
\ket{\psi_j}=e^{i\theta_j}W(z_j)\ket{\phi}
\end{equation}
for some $\theta_j\in\R$ and $z_j\in \AhA$ (hence $\psi_j=e^{i\theta_j}\pi(z_j)\phi$; cf. \eqref{eq tfshift} and \eqref{eq weylop}), and that $\mathcal{K}$ is a coset in $\AhA$ of the maximal compact open isotropic subgroup $K\subset\AhA$ given by
\begin{equation}\label{eq 24gbis}
 K=\{z\in\AhA: V_\phi \phi(z)\not=0\}.
 \end{equation}
We also know that the $\ket{\psi}_j$'s are pairwise orthogonal. Hence, if there exist two indices $j,k$ with $j\not=k$ in the sum \eqref{eq rho decomp}, we have 
\[
|V_\phi \phi(z_j-z_k)|=|\langle \phi|W(-z_k) W(z_j)| \phi\rangle|=|\langle \psi_k|\psi_j\rangle|=0,
\]
that is, $z_j$ and $z_k$ belong to different cosets of $K$. By \eqref{eq 24g}, \eqref{eq 24gbis} and the covariance property \eqref{eq 2.3} we deduce that $z_j+K=\{z\in\AhA: V_\phi \psi_j(z)\not=0\}$ and $z_k+K=\{z\in\AhA: V_\phi \psi_k(z)\not=0\}$, which contradicts \eqref{eq vphipsik}. In summary, $\rho$ is a pure state and the desired properties of $\ket{\phi}$ and $\rho$ have been proved.

$(iii)\Longrightarrow (i)$ It follows from Theorem \ref{teo nic} (ii) that $|V_\phi \phi|=\chi_K$ for some maximal compact open isotropic subgroup $K\subset \AhA$, hence with $|K|=1$ by Theorem \ref{teo iso sub}. By the covariance property \eqref{eq 2.3} we see that 
\[
u_{\phi,\rho}=|V_\phi \psi|^2=\chi_\mathcal{K}
\]
 for some subset $\mathcal{K}\subset \AhA$ -- a coset of $K$ in $\AhA$ -- of measure $|\mathcal{K}|=1$. As a consequence \[
 \cE_G(\phi,\rho)= \int_{\AhA} G(|V_\phi\psi(x,\xi)|^2) \, dx\, d\xi=G(1)|\mathcal{K}|=G(1).
 \]
 \end{proof}

 \section{Further results for finite-dimensional systems}\label{sec finite dim}
 In this section we suppose that $A$ is a finite Abelian group, say of cardinality 
 \[
 \#A=N.
 \]
 Hence $\cH\simeq L^2(A)$ has dimension $N$. We prove some results for the Wehrl entropy in \eqref{eq wehrl entropy} that are, at least in part, specific for this finite-dimensional setting. Observe that now we have (cf. Remark \ref{rem osservazioni su mainteo} 2))
\begin{equation}\label{eq egnuova}
\mathcal{E}_G(\phi,\rho)=\frac{1}{N}\sum_{z\in\AhA} G(u_{\phi,\rho}(z)).
\end{equation}
\subsection{Berezin-Lieb inequality} 
Let $\rho$ be a density operator in $\cH$. Let  \begin{equation}\label{eq rho spec}
    \rho=\sum_{j=1}^N p_j\ket{\psi_j}\bra{\psi_j}
    \end{equation}
    be a spectral decomposition of $\rho$. Hence $p_j\geq0$ for $j=1,\ldots,N$,  $\sum_{j=1}^N p_j=1$ and the $\ket{\psi_j}$'s define now an \textit{orthonormal basis} of $\cH$. 
    Given a function $G:[0,1]\to\R$, we define the operator 
\[
G(\rho):=\sum_{j=1}^N G(p_j)\ket{\psi_j}\bra{\psi_j}.
\]
Hence 
\[
\Tr G(\rho)=\sum_{j=1}^N G(p_j). 
\]
If $G(\tau)=-\tau\log \tau$ for $0<\tau\leq 1$ and $G(0)=0$, $\Tr G(\rho)$ is the von Neumann entropy of $\rho$; see, e.g., \cite[Section 5.2]{holevo2019quantum}. 

\begin{theorem}\label{teo von neum}
Let $G:[0,1]\to\R$ be a concave function. For every $\ket{\phi}\in\cH$, with $\|\ket{\phi}\|=1$, and density operator $\rho$, we have
\begin{equation}\label{eq von neu}
\cE_G(\phi,\rho)\geq \Tr G(\rho).
\end{equation}
Equality occurs in \eqref{eq von neu} for every $\ket{\phi}\in\cH$, with $\|\ket{\phi}\|=1$ if $\rho= N^{-1} I$ (the so-called chaotic state). 

Moreover, equality occurs in \eqref{eq von neu}  if $\ket{\phi}$ is an $S$-state vector (Definition \ref{def state}), $\rho$ has the form in \eqref{eq rho spec} with
\[
\ket{\psi_j}=e^{i\theta_j}W(z_j)\ket{\phi}
\]
for some $\theta_j\in\R$, and some $z_j\in \AhA$,  $j=1,\ldots,N$, belonging each to a different coset in $\AhA$ of the maximal isotropic subgroup $K:=\{z\in\AhA: V_\phi \phi(z)\not=0\}$ (hence $\#K=N$).

If $G$ is strictly concave, equality occurs in \eqref{eq von neu} for some pair $\ket{\phi}, \rho$, with $\|\ket{\phi}\|=1$ and $\rho$ having distinct eigenvalues, only if $\ket{\phi}$ and $\rho$ have the above form.
\end{theorem}
% We will see (Proposition \ref{}) that the vectors $\ket{\psi_j}$ appearing in Theorem \ref{} 3) are common eigenvectors of a maximal family of pairwise commuting Weyl-Heisenberg operators, namely the $\pi(z)$'s, with $z\in K_h$. 

%If one (and therefore every) condition above holds true, the characteristic functions $\tilde{\phi}$ and $\tilde{\psi}$ have the same support, say $K\subset\AhA$. Moreover $K$ is a maximal isotropic subgroup of $\AhA$ and both $\tilde{\phi}$ and $\tilde{\psi}$ are subcharacters of second degree of $\AhA$.  

\begin{proof}[Proof of Theorem \ref{teo von neum}]
Form \eqref{eq rho spec} we have (cf. \eqref{eq husimi decomp})
\begin{equation}
    \label{eq husimi decomp 2}
   u_{\phi,\rho}(z)=\sum_{j=1}^N p_j|V_\phi \psi_j(z)|^2\qquad z\in\AhA.
\end{equation}
Hence
\begin{align*}
\mathcal{E}_G(\phi,\rho)&=\frac{1}{N}\sum_{z\in\AhA} G\Big(\sum_{j=1}^N p_j|V_\phi \psi_j(z)|^2\Big)\\
&\geq \frac{1}{N}\sum_{z\in\AhA}\sum_{j=1}^N G(p_j) |V_\phi \psi_j(z)|^2\\
&=\sum_{j=1}^N G(p_j),
\end{align*}
where the inequality follows from the concavity of $G$ and the fact that 
\begin{equation}\label{eq 26}
\sum_{j=1}^N |V_\phi\psi_j(z)|^2=\sum_{j=1}^N |\langle \psi_j|W(z)|\phi\rangle|^2=
\|W(z)|\phi\rangle \|^2=1\quad z\in\AhA,
\end{equation} 
 whereas the last equality is a consequence of \eqref{parseval}, that here reads 
 \begin{equation}\label{parseval2}
\frac{1}{N}\sum_{z\in\AhA}|V_\phi \psi_j(z)|^2=1\qquad j=1,\ldots,N.
 \end{equation}

It is clear that equality occurs in \eqref{eq von neu} for every $\ket{\phi}$, with $\|\ket{\phi}\|=1$ if $\rho= N^{-1} I$, by \eqref{eq 26}.

From the above argument we see that, more generally, equality occurs in \eqref{eq von neu} if and only if, for every $z\in\AhA$,
\begin{equation}\label{eq 4g}
G\Big(\sum_{j=1}^N p_j|V_\phi \psi_j(z)|^2\Big)=\sum_{j=1}^N G(p_j) |V_\phi \psi_j(z)|^2.
\end{equation}
If $G$ is strictly concave and $p_j\not=p_k$ for $j,k\in \{1,\ldots,N\}$, $j\not=k$, this happens if and only if, for every $z\in \AhA$ there exists (one and) only one $j\in\{1,\ldots,N\}$ such that $V_\phi \psi_j(z)\not=0$. That is, the functions $V_\phi \psi_j$, $j=1,\ldots,N$, have disjoint supports.
By \eqref{eq cs} and \eqref{parseval2} such supports have cardinality at least $N$, and therefore they must have all cardinality $N$. By Theorem \ref{teo nic} (iii) this implies that $\ket{\phi}$ is an $S$-state vector (Definition \ref{def state}) and there exist $\theta_j\in\R$, $z_j\in \AhA$, with $j=1,\ldots,N$ such that 
\[
\ket{\psi_j}=e^{i\theta_j}W(z_j)\ket{\phi}.
\]
Moreover,  it follows from Theorem \ref{teo nic} (ii), that the set $K$ in the statement is a maximal isotropic subgroup of $\AhA$.  
 Since the $\ket{\psi_j}$'s are pairwise orthogonal, arguing as in the proof of Theorem \ref{teo wehrl} we see that the $z_j$'s belong to different cosets of $K$. 
 
 Conversely, it is clear, essentially by revering this latter reasoning, that for  $\ket{\phi}$ and $\ket{\psi_j}$ as in the statement, the vectors $\ket{\psi_j}$ define an orthonormal basis of $\cH$, and the functions $V_\phi \psi_j$ have disjoint supports. Hence \eqref{eq 4g} is satisfied for every $z\in\AhA$ regardless of whether $G$ is strictly concave and whether the $p_j$'s are all distinct (indeed, the sums in \eqref{eq 4g} contain only one summand).  
\end{proof}
%\begin{remark}
%Notice that the so-called chaotic state $\rho=N^{-1} I$ can be written in the form \eqref{eq rho spec} with $p_j=N^{-1}$ and, for example,  $\ket{\psi_j}=W(z_j)\ket{\phi}$ for $j=1,\ldots, N$, with $\phi(x)\equiv N^{-1/2}$. Hence for such $\rho$, equality equality occurs  (In this case $V_\phi \phi=\chi_{K}$, with $K=A\times\{0\}$. 
%\end{remark}
\subsection{Maximizing the Wehrl entropy}
We consider the problem of maximizing the Wehrl entropy in \eqref{eq wehrl entropy} when $A$ is a finite Abelian group, say of cardinality $N$. 

We recall from \cite[Section 12.3]{holevo2019quantum} that the {\it characteristic function} of a density operator $\rho$ is defined by
\begin{equation}
\tilde{\rho}(z):=\Tr \rho W(z)\qquad z\in\AhA.
\end{equation}
When $\rho=\ket{\psi}\bra{\psi}$ is a rank one projector, we simply write $\tilde{\psi}(z)$ for $\tilde{\rho}(z)$.

We also consider the Fourier transform on $\AhA$, defined by 
\[
\mathcal{F} f(\xi',x')=\frac{1}{N}\sum_{(x,\xi)\in \AhA}\overline{\xi'(x)\xi(x')}f(x,\xi) \qquad (\xi',x')\in \hA\times A. 
\]
We will need the following useful formula for the Fourier transform of the Husimi function. Related formulas have appeared in various contexts (such as mathematical signal processing \cite{alaifari22}) but with different notation and terminology. 
\begin{proposition}\label{pro pro6.2}
    Let $\ket{\phi}\in\cH$, with $\|\ket{\phi}\|=1$  and let $\rho$ be a density operator in $\cH$. We have 
    \[
    \mathcal{F} u_{\phi,\rho}(\xi,-x)=\tilde{\phi}(x,\xi)\overline{\tilde{\rho}(x,\xi)}\qquad (x,\xi)\in \AhA.  
    \]
\end{proposition}
\begin{proof}
It suffices to prove the desired result when $\rho$ is a rank one projector, therefore $\rho=\ket{\psi}\bra{\psi}$. In that case, we have (cf. \eqref{eq husimi decomp 2}) 
\[
u_{\phi,\rho}(z)=|V_\phi \psi(z)|^2.
\]
By direct inspection one sees that 
\[
\tilde{\phi}=\overline{V_\phi \phi}\qquad  \tilde{\rho}=\overline{V_\psi \psi}.
\]
As a consequence, we are reduced to prove that 
\[
\mathcal{F} |V_\phi \psi|^2 (\xi,-x)= \overline{V_\phi\phi (x,\xi)} V_\psi \psi(x,\xi)\qquad (x,\xi)\in \AhA.  
\]
This formula was proved in \cite[Lemma 2.2]{alaifari22}
when $A=\bZ_d$. The proof given there can be easily adapted to any finite Abelian group (in fact, to any LCA group). 
\end{proof}

We also need the following elementary lemma. 
\begin{lemma}\label{lem max f}
Let $G:[0,1]\to\R$ be a concave function and $N\geq1$ be an integer. Let 
\[
\mathcal{P}:=\Big\{t=(t_1,\ldots, t_{N^2})\in[0,1]^{N^2}:\ \sum_{j=1}^{N^2} t_j=N\Big\},
\]
and 
\[
f(t):=\frac{1}{N}\sum_{j=1}^{N^2} G(t_j)\qquad t\in\mathcal{P}.
\]
Then $f$ attains its maximum value on $\mathcal{P}$ at $\overline{t}:=(N^{-1},\ldots,N^{-1})$. 

If $G$ is strictly concave, then $\overline{t}$ is the unique maximum point.     
\end{lemma}
\begin{proof}
   For every $t=(t_1,\ldots,t_{N^2})\in\mathcal{P}$ and every permutation \[
   \sigma:\{1,\ldots,N^2\}\to \{1,\ldots,N^2\},
   \]
   consider the point $t_\sigma:=(t_{\sigma(1)},\ldots,t_{\sigma(N^2)})\in\mathcal{P}$. Using the fact that $\sum_{j=1}^{N^2} t_j=N$ it is easy to check that 
   \[
   \overline{t}=\frac{1}{(N^2)!}\sum_{\sigma} t_\sigma.  
   \]
   Indeed, each component of the vector $\sum_{\sigma} t_\sigma$ is equal to 
   \[
(N^2-1)!t_1+(N^2-1)!t_2+\ldots+(N^2-1)! t_{N^2}=(N^2-1)! N.
   \]
   Since $G$ is concave,   $f$ is concave as well  and therefore 
   \[
f(\overline{t})\geq\frac{1}{(N^2)!}\sum_\sigma f(t_\sigma)=\frac{1}{(N^2)!}\sum_\sigma f(t)=f(t). 
   \]
The last part of the statement is clear, because $f$ is strictly concave whenever $G$ is strictly concave.
\end{proof}
We can now state the desired characterization of the maximizers of the Wehrl entropy. For a function $f:\AhA\to\mathbb{C}$ we set ${\rm supp}\, f=\{z\in \AhA:\ f(z)\not=0\}$. 
\begin{theorem}\label{teo wehrl max}
    Let $G:[0,1]\to\R$ be a concave function. Let $\ket{\phi}\in\cH$, with $\|\ket{\phi}\|=1$  and let $\rho$ be a density operator in $\cH$. Then 
    \begin{equation}\label{eq 28}
        \cE_G(\phi,\rho)\leq N\, G(1/N).
\end{equation}
Moreover, if $G$ is strictly concave, equality occurs in \eqref{eq 28} if and only if
\begin{equation}\label{eq 29}
{\rm supp}\,\tilde{\phi}\cap {\rm supp}\,\tilde{\psi} =\{0\}. 
\end{equation}
\end{theorem}
\begin{proof}
    The upper bound  \eqref{eq 28} follows from \eqref{eq egnuova} and Lemma \ref{lem max f}. 

    Suppose that $G$ is strictly concave. Then it follows from Lemma \ref{lem max f} that equality occurs in \eqref{eq 28} if and only if 
    \[
u_{\phi,\rho}(z)=\frac{1}{N}\qquad z\in\AhA.
    \] Taking the Fourier transform of both sides, we see that this is equivalent to 
    \[
    \mathcal{F}u_{\phi,\rho}(\xi,x)=\begin{cases}1&(\xi,x)=0\\
    0&(\xi,x)\not=0.
    \end{cases}
    \]
    Using Proposition \ref{pro pro6.2}, this amounts to \eqref{eq 29}, because  $\tilde{\rho}(0)=\Tr \rho=1$ and, similarly, $\tilde{\phi}(0)=1$.
\end{proof}
\begin{remark}
 If $G(\tau)=-\tau \log \tau$ for $0<\tau\leq 1$ and $G(0)=0$, as already observed, $\Tr G(\rho)$ is the von Neumann entropy of $\rho$. In this case, combining Theorems \ref{teo von neum} and \ref{teo wehrl max} we obtain that 
    \[
\Tr G(\rho)\leq \cE_G(\phi,\rho)\leq \log N.
    \]
    It is well known and easy to check that $\Tr G(\rho)=\log N$ if and only if $\rho=N^{-1} I$. Theorems \ref{teo von neum} and \ref{teo wehrl max} provide information on the cases where the two above inequalities hold separately as an equality.  
\end{remark}

%\section*{Acknowledgments}

%\bibliographystyle{abbrv} 
%\bibliography{biblio}

\end{document}